\documentclass[twocolumn]{IEEEtran}
\usepackage[final]{graphicx}
\usepackage{amsthm}
\usepackage{amsmath,epsfig,amssymb,verbatim,amsopn,cite,multirow}
\usepackage{balance}
\usepackage{multirow}
\usepackage{footnote}
\usepackage{stfloats}
\usepackage{algorithm}
\usepackage{algorithmic}
\usepackage[usenames,dvipsnames]{color}
\usepackage[all]{xy}  
\usepackage{url}
\usepackage{subcaption}
\usepackage{graphicx}
\usepackage{array}
\usepackage{threeparttable}

\usepackage[nodisplayskipstretch]{setspace}
  {\proof}{\proofend}
\newtheorem{proposition}{Proposition}

\newcommand{\qA}{{\bf A}}

\title{Joint AP Selection and Power Allocation for Unicast-Multicast Cell-Free Massive MIMO}
\author{Mustafa S. Abbas, Zahra Mobini,~\IEEEmembership{Member,~IEEE,} Hien Quoc Ngo,~\IEEEmembership{Fellow,~IEEE},\\~Hyundong Shin,~\IEEEmembership{Fellow, IEEE,} and Michail Matthaiou,~\IEEEmembership{Fellow,~IEEE}
\thanks{This work was supported by the U.K. Engineering and Physical Sciences Research Council (EPSRC) (grants No. EP/X04047X/1 and EP/X040569/1). The work of Z.~Mobini and  H.~Q.~Ngo was supported by the U.K. Research and Innovation Future Leaders Fellowships under Grant MR/X010635/1. The work of M. Matthaiou was supported by the European Research Council (ERC) under the European Union’s Horizon 2020 research and innovation programme (grant agreement No. 101001331). The work of H. Q. Ngo and M. Matthaiou was also supported by a research grant from the Department for the Economy
Northern Ireland under the US-Ireland R\&D Partnership Programme. The work of  H. Shin was supported by the National Research Foundation of Korea (NRF) grant funded by the Korean government (MSIT) under RS-2025-00556064 and by the MSIT (Ministry of Science and ICT), Korea, under the ITRC (Information Technology Research Center) support program (IITP-2025-2021-0-02046) supervised by the IITP (Institute for Information \& Communications Technology Planning \& Evaluation). (\textit{Corresponding authors: Hien Quoc Ngo; Hyundong Shin}.)
}
\thanks{M. S. Abbas,  H.~Q.~Ngo, and  M. Matthaiou are with the Centre for Wireless Innovation (CWI), Queen’s University Belfast, BT3 9DT Belfast, U.K. (e-mails: \{malkhadhrawee01,  hien.ngo, m.matthaiou\}@qub.ac.uk). H.~Q.~Ngo is also with the Department of Electronic Engineering, Kyung Hee University, Yongin-si, Gyeonggi-do 17104, Republic of Korea.}
\thanks{Z.~Mobini is with the Centre for Wireless Innovation (CWI), Queen’s University Belfast, BT3 9DT Belfast, U.K., and also   with the Department of Electrical and Electronic Engineering, The University of Manchester, Manchester M13 9PL, U.K. (e-mail: zahra.mobini@manchester.ac.uk).}

\thanks{H.~Shin is with the Department of Electronics and Information Convergence Engineering,
Kyung Hee University,
1732 Deogyeong-daero, Giheung-gu,
Yongin-si, Gyeonggi-do 17104, Republic of Korea
(e-mail: hshin@khu.ac.kr).}

\thanks{Parts of this paper were presented at the 2024 IEEE GLOBECOM~\cite{Mustafa:2024:globe}.}
}
\newcommand{\MR}{\mathrm{MR}}
\newcommand{\ZF}{\mathrm{ZF}}

\begin{document}


\maketitle

\begin{abstract}
Joint unicast and multicast transmissions are becoming increasingly important in practical wireless systems, such as Internet of Things  networks. This paper investigates a cell-free massive multiple-input multiple-output system that simultaneously supports both transmission types, with multicast serving multiple groups.  Exact closed-form expressions for the achievable downlink spectral efficiency (SE) of both unicast and multicast users are derived for zero-forcing  and maximum ratio precoding designs. Accordingly, a weighted sum SE (SSE) maximization problem is formulated to jointly optimize the access point (AP) selection and power allocation. The optimization framework accounts for practical constraints, including the maximum transmit power per AP, fronthaul capacity limitations between APs and the central processing unit, 
and quality-of-service requirements for all users.
The resulting non-convex optimization problem is reformulated into a tractable structure, and an accelerated projected gradient (APG)-based algorithm is developed to efficiently obtain near-optimal solutions. As a performance benchmark, a successive convex approximation (SCA)-based algorithm is also implemented.
Simulation results demonstrate that the proposed joint optimization approach significantly enhances the SSE across various system setups and precoding strategies. In particular, the APG-based algorithm achieves substantial complexity reduction while maintaining competitive performance, making it well-suited for large-scale practical deployments.

\end{abstract}

\section{Introduction}
Cell-free massive multiple-input multiple-output (CF-mMIMO) technology represents a significant advancement for next-generation wireless systems, including large-scale internet of things (IoT) deployments and diverse use cases such as smart cities,   environmental monitoring/surveillance,   and remote healthcare systems\cite{ngo2024ultra,Matthaiou:COMMag:2021,Mohammadi:Proc:2024,Zahra:IOT:2024}. 
In this architecture, numerous access points (APs) are distributed across a wide coverage area and collaboratively operate on the same time-frequency resources via time-division duplexing (TDD), enabling service provision to a large number of users without being constrained by traditional cell boundaries. The CF-mMIMO architecture employs backhaul links to connect central processing units (CPUs) and fronthaul links to connect APs and CPUs. CF-mMIMO offers several key advantages, including channel hardening, favorable propagation conditions, and enhanced macro-diversity \cite{Mohammadi:Proc:2024}. As a result, it enables reliable connectivity over extensive areas while achieving high energy efficiency (EE) and spectral efficiency (SE) \cite{ngo2024ultra,Hien:JWCOM:2017}. These benefits have recently attracted considerable research interest.
In \cite{Ngo:IEEE_G_Net:2018}, the authors investigated unicast CF-mMIMO systems with multi-antenna APs and proposed a successive convex approximation (SCA) method to optimize the total EE. Additionally, they introduced an AP selection strategy to further improve the system performance. In \cite{Farooq:TCOM:2021}, an accelerated projected gradient (APG) method was proposed to maximize various system-wide utility functions. The APG approach was also explored in \cite{zhang2022cell} for optimizing the power allocation in CF-mMIMO systems with simultaneous wireless information and power transfer.

Meanwhile, recent statistics indicate that the rapid growth of IoT applications for widespread connectivity is expected to result in over $75.4$ billion IoT-connected devices relying on wireless networks. To support this massive scale, multicast transmission has become increasingly essential in IoT networks, as numerous devices require the same updates or control messages. By reducing redundant transmissions, multicast communication ensures scalable and energy-efficient connectivity. Currently, multigroup multicast, alongside traditional unicast, accounts for $53.72\%$ of overall network traffic, highlighting the rising demand for multicast applications in modern wireless services\cite{Li:JIOT:2024,Zhou:JIOT:2025}. In response to this growing need, numerous studies have explored advanced multicast transmission strategies \cite{Min:TSP:2020,Alejandro:OJCOMS:2024,Farooq:EUSIPCO:2021,Doan:LCOMM:2017}. For instance, Dong \emph{et. al.,} in \cite{Min:TSP:2020}  developed an efficient algorithm to compute the multicast beamformer in a single-cell massive MIMO system using the SCA method, addressing both quality of service (QoS) and max-min fairness (MMF) optimization. A comparative analysis of the bisection and APG methods for MMF-based power allocation was presented in \cite{Farooq:EUSIPCO:2021}. In \cite{Alejandro:OJCOMS:2024}, a subgroup-centric multicast strategy was introduced for CF-mMIMO based on spatial channel characteristics, integrating centralized improved partial MMSE (IP-MMSE) processing with distributed conjugate beamforming. The performance of multigroup multicasting in CF-mMIMO systems under short-term power constraints on the beamformers was investigated in \cite{Doan:LCOMM:2017}. Additionally, \cite{Zhou:JSYST:2022} investigated the effects of multi-antenna users and low-resolution analog-to-digital and digital-to-analog converters (ADCs/DACs) on the multigroup multicast performance in CF-mMIMO systems.
However, these works focus exclusively on either unicast or multicast scenarios. In practical deployments—especially with the proliferation of massive access use cases—wireless systems should increasingly support both unicast and multicast services. Joint unicast-multicast transmission is particularly valuable in the evolution from fifth-generation (5G) wireless networks and beyond, as it enables more efficient utilization of spectrum resources in both access and backhaul slices \cite{Natalia:TVT:2024}. By accommodating both shared content and user-specific demands, these joint transmissions promise improved SE and resource management.

In the context of CF-mMIMO systems, resource allocation optimization has received significant research attention, with numerous studies focusing on strategies to efficiently manage the power budget and user scheduling to enhance the system performance. However, the majority of these works addressed either unicast or multicast transmissions separately. For instance, the studies in\cite{Hien:JWCOM:2017,Ngo:IEEE_G_Net:2018,Farooq:TCOM:2021,zhang2022cell,Sutton:TVT:2021,Tuan:TCOMM:2022} focused on unicast users, while those in \cite{Min:TSP:2020,Alejandro:OJCOMS:2024,Farooq:EUSIPCO:2021,Doan:LCOMM:2017,Zhou:JSYST:2022,Gouda:TWC:2024} considered only multicast scenarios.  Although each traffic type has been studied in isolation, the work in \cite{Sadeghi:TWC:2018} considered both unicast and multicast services but treats them separately, without a unified optimization framework. The development of efficient resource allocation strategies for joint unicast-multicast transmission in CF-mMIMO systems is therefore of considerable practical importance. Nevertheless, these systems introduce unique challenges due to increased problem dimensionality and complex interference patterns. The presence of both unicast and multicast users often results in multi-objective optimization problems (MOOPs) that are computationally intensive. Moreover, pilot contamination and inter-group interference are amplified with the number of multicast users.
\begin{table*}[t]
\centering
\caption{Our contributions compared to the state-of-the-art in joint unicast-multicast transmission}
\begin{tabular}{|c|c|c|c|c|c|c|}
\hline
\textbf{Feature}                       &\cite{Zhe:TCOM:2023} & \cite{Sadeghi:WCOM:2018} & \cite{Shadi:TSP:2022}& \cite{Li:TVT:2022}  &\cite{Tan:ISNCC:2021}   & \textbf{Our work} \\ 
\hline 
CF-mMIMO         & & &  &\checkmark  &\checkmark  &\checkmark                \\ 
\hline
 Closed-form SE expression for joint unicast-multicast  & &\checkmark  &\checkmark$^a$ & & \checkmark$^a$ & \checkmark                
 \\ \hline
Maximize SSE 
&\checkmark$^b$ &\checkmark$^b$ & &\checkmark$^b$   & 
 &\checkmark               \\ \hline
 Joint AP selection and power allocation 
& & & &   & 
 &\checkmark               \\ \hline
APG-based optimization approach & & &  &  &  &  \checkmark 
\\ \hline
\end{tabular}
\begin{tablenotes}
   \item \hspace{15em}$^a$only MR, $^b$maximized SSE for unicast only
\end{tablenotes}
\label{table1}
\end{table*}

Several works have investigated joint unicast-multicast transmission using various optimization frameworks \cite{Zhe:TCOM:2023,Sadeghi:WCOM:2018,Shadi:TSP:2022,Li:TVT:2022,Tan:ISNCC:2021}. For instance, \cite{Zhe:TCOM:2023} proposed a graph neural network-based beamforming scheme for a   multiple-input single-output (MISO) system under imperfect channel state information (CSI). In \cite{Sadeghi:WCOM:2018}, the authors formulated a MOOP to jointly maximize the SSE of unicast users and the MMF of multicast users using Pareto boundary techniques in a single-cell massive MIMO setup. Similarly, \cite{Shadi:TSP:2022} developed a fast algorithm using the alternating direction method of multipliers to minimize the total transmit power, decomposing the joint problem into unicast and multicast subproblems while ensuring QoS requirements. The work in \cite{Li:TVT:2022} considered a CF-mMIMO system for joint unicast-multicast communication, using deep learning and non-dominated sorting genetic algorithm II to optimize the system performance. Lastly, \cite{Tan:ISNCC:2021} focused on the EE of layered-division multiplexing for joint transmission using SCA and Dinkelbach’s method.

In parallel, user association and AP selection in CF-mMIMO have attracted attention for their potential to reduce signaling overhead and enhance scalability \cite{Hao:2023:EUSIPCO, hao2024joint, Yasseen:TWC:2023}. Deactivating APs with negligible impact on system performance can also contribute to significant energy savings \cite{Yige:TVT:2024, Thang:ICC:2020}, while \cite{Ngo:IEEE_G_Net:2018} highlighted that AP selection is a key method to reduce fronthaul and backhaul signaling, which are key limitations in CF-mMIMO systems. Although \cite{hao2024joint} presented a joint power allocation and AP selection scheme, their study was limited to unicast-only systems. Despite its growing importance, joint AP selection and power allocation in joint unicast-multicast CF-mMIMO systems remain underexplored, particularly under realistic conditions with  fronthaul constraints and heterogeneous QoS demands. This gap underscores the need for a unified and efficient approach to jointly optimize power allocation and AP selection in such systems.

Motivated by the aforementioned considerations, this paper investigates a joint unicast-multicast CF-mMIMO system and proposes a novel framework for joint power control and user association, leveraging the APG algorithm to enhance the SSE. The APG method is chosen due to its computational efficiency and low memory requirements, which make it particularly well-suited for large-scale wireless network deployments. To benchmark the effectiveness of the proposed APG-based scheme, we also implement the SCA algorithm. The main contributions of this paper are as follows:

\begin {itemize}
\item  We derive closed-form SE expressions for joint unicast-multicast CF-mMIMO systems employing maximum-ratio (MR) and zero-forcing (ZF) precoding schemes, using the use-and-then-forget bounding technique. These closed-form expressions incorporate the effects of imperfect CSI and power control. In contrast to the approximation in \cite{Li:TVT:2022}, our result is exact and for finite antenna arrays, thus better reflecting practical system conditions. 

\item We then formulate a weighted SSE maximization problem that jointly addresses power allocation and AP selection, considering  the per-AP power constraints, as well as user QoS requirements and  fronthaul limitations. We transform the complex binary non-convex optimization problem into a more tractable form involving only continuous variables. We then develop two solution methods based on the  APG  approach and  SCA. The APG method offers low computational complexity, whereas the SCA method achieves near-optimal performance at the cost of higher computational complexity. 

\item Our numerical results demonstrate that the proposed APG-based joint AP selection and power allocation scheme significantly enhances the SSE performance in joint unicast-multicast CF-mMIMO systems. Under specific SE and fronthaul constraints, the proposed method, employing ZF precoding, achieves up to an order-of-magnitude improvement in the SSE compared to equal power allocation combined with random AP selection heuristics. Furthermore, the results confirm that our APG-based approach enables efficient implementation of joint AP selection and power allocation in joint unicast-multicast CF-mMIMO systems, delivering performance comparable to SCA-based methods but with substantially lower computational complexity.  
\end {itemize}
Table \ref{table1} provides a  comparison of our paper’s contributions with those of related studies in the literature.

\textit{Notation:} The superscripts $(\cdot)^T$, $(\cdot)^\ast$, and $(\cdot)^H$ denote the transpose, conjugate, and conjugate-transpose, respectively. The symbols  
$\mathbf{I}_n$ and  $\mathbb{E}\{\cdot\}$ stand for the $n\times n$ identity matrix, and the statistical expectation, respectively. 
Finally, a  circular symmetric complex Gaussian variable having variance $\sigma^2$ is denoted by $\mathcal{CN}(0,\sigma^2)$. 

\section{System Model}\label{sysmod}
We consider a CF-mMIMO system with joint unicast and multi-group multicast transmissions, as illustrated in Fig.~\ref{fig:1}. The system consists of $N$ APs, each equipped with $L$ antennas, that serve simultaneously $U$ unicast users and $M$ multicast groups, where the $m$-th group includes $K_m$ users. The sets of $N$ APs, $U$ unicast users, $M$ multicast groups, and $K_m$ users in the $m$-th unicast group are denoted by $\mathcal{N}$, $\mathcal{U}$, $\mathcal{M}$, and $\mathcal{K}_m$, respectively. 
The  channel vector  between the $u$-th unicast user, $u \in \mathcal{U}$, and the $n$-th AP, $n \in \mathcal{N}$, is 
 \begin{equation} \label{eq:cnu}
 \mathbf c_{n,u} = \beta_{n,u}^{1/2} \mathbf h_{n,u}\in \mathbb{C}^{L\times1}.
\end{equation} 
Moreover, the channel between the $k_m$-th multicast user, $k_m \in \mathcal{K}_m$, of the $m$-th multicast group, $m \in \mathcal{M}$, and the $n$-th AP is
\begin{equation} \label{eq:tnmk}
 \mathbf t_{n,m,k} = \bar{\beta}_{n,m,k}^{1/2} \mathbf h_{n,m,k}\in \mathbb{C}^{L\times1}.
\end{equation}
In this context, $\beta_{n,u}$ and $\bar{\beta}_{n,m,k}$ denote the large-scale fading coefficients. Additionally, $\mathbf h_{n,u} \sim \mathcal{CN} (\mathbf 0,\mathbf I_{L})$ and $\mathbf h_{n,m,k} \sim \mathcal{CN} (\mathbf 0,\mathbf I_{L})$ represent the small-scale fading vectors.

\begin{figure}[t]
    \centering
    \includegraphics[width=0.7\linewidth]{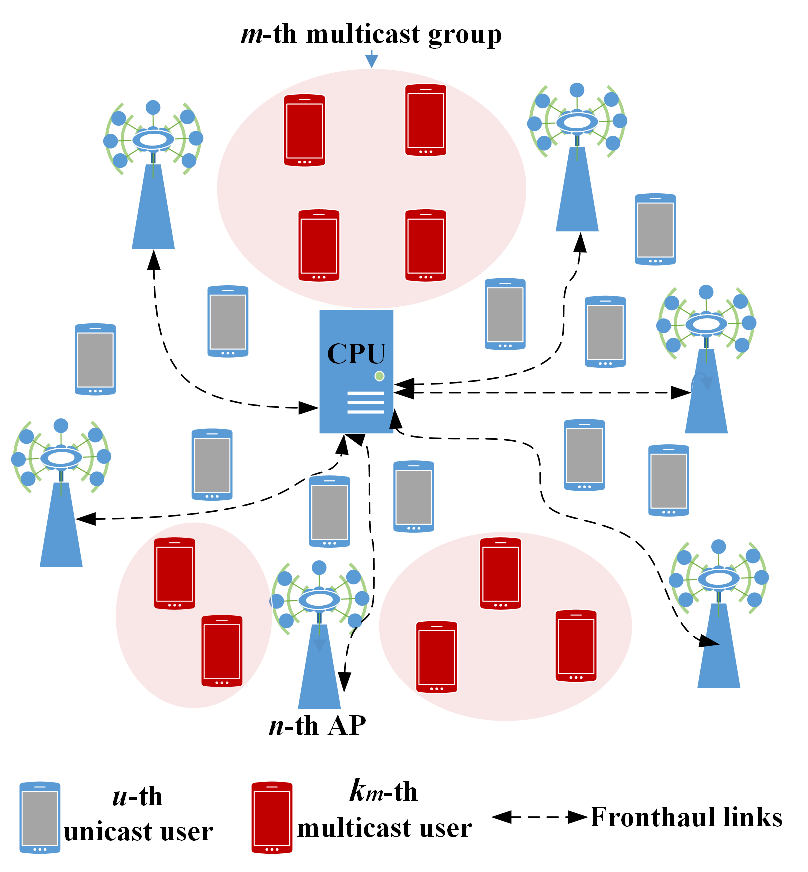}
    \caption{ CF-mMIMO system with joint unicast and multi-group multicast transmissions. }
    \label{fig:1}
\end{figure}

\subsection{Uplink Training}\label{up}
The system is assumed to work under the reciprocity-based TDD protocol, where the channels remain unchanged during a coherence interval $T$. The APs acquire the CSI through uplink training. We assume that the pilots dedicated to the unicast users are orthogonal. However, we assign a shared pilot to all the users in each multicast group. This is because in practice the length of the coherence interval is limited and each $\mathcal{K}_m$ group of users desires the same data \cite{yang:SPAWC:2013}. Therefore, our considered CF-mMIMO system requires $U + M$ orthogonal pilots. Let $ \boldsymbol{\phi}_{u}\in\mathbb{C}^{\tau  \times 1}$,  $\|\boldsymbol{\phi}_{u} \|^2 =1$, be the pilot sequence sent by the $u$-th unicast user, and $ \boldsymbol{\varphi}_{m}\in\mathbb{C}^{\tau  \times 1}$,  $\|\boldsymbol{\varphi}_{m} \|^2 =1$, be the pilot assigned to   all the multicast users in the $m$-th multicast group, while  the pilot length $\tau$ satisfies the condition $U+M \leq \tau \leq T$. Thus, we have $\boldsymbol{\phi}_{u}^H \boldsymbol{\phi}_{u'}=0$ for $u \neq u'$, $\boldsymbol{\phi}_{u}^H\boldsymbol{\varphi}_{m} =0$ and $\boldsymbol{\varphi}_{m}^H \boldsymbol{\varphi}_{m'} =0$ for $m \neq m'$.
 The received signal at the $n$-th AP during uplink training can be written as
\begin{align} \label{eq:ynp}
    \mathbf Y_{n,p}=&\sqrt{\tau p_{\text{ul}}}\sum\nolimits_{u \in \mathcal{U}}  \mathbf c_{n,u} \boldsymbol{\phi}_{u}^H\nonumber\\
    & + \sqrt{\tau p_{\text{ul}}} \sum\nolimits_{m \in \mathcal{M}}  \sum\nolimits_{k   \in \mathcal{K}_m}  \mathbf t_{n,m,k} \boldsymbol{\varphi}_{m}^H + \mathbf \Psi_{n,p},
\end{align}
where $\mathbf \Psi_{n,p}\in\mathbb{C}^{L  \times \tau}$ represents the additive white Gaussian noise and $p_{\text{ul}}$ denotes the uplink transmit power.
To estimate $\mathbf c_{n,u}$, we project the received signal $\mathbf Y_{n,p}$ onto the pilot sequence corresponding to the $u$-th unicast user, $\boldsymbol{\phi}_{u}$, as
\begin{equation} \label{eq:ynpu}
    \mathbf{\Check{y}}_{n,p,u} = \mathbf Y_{n,p} \boldsymbol{\phi}_{u} = \sqrt{\tau p_{\text{ul}}}  \mathbf c_{n,u}  + \boldsymbol{\psi'}_{n,p},
\end{equation}
where $\boldsymbol{\psi'}_{n,p}= \mathbf \Psi_{n,p} \boldsymbol{\phi}_{u}\sim \mathcal{CN} (\mathbf 0,\mathbf I_{L})$. Each AP estimates the user channel locally, thereby minimizing the backhaul signaling. The MMSE estimate of $\mathbf c_{n,u} \in \mathbb{C}^{L\times 1 }$ is 
\begin{align}
     \label{eq:cnuh}
 \mathbf {\hat{c}}_{n,u} & = \mathbb{E}\big\{\mathbf c_{n,u} \mathbf{\Check{y}}_{n,p,u}^H\big\} \Big(\mathbb{E} \big\{\mathbf{\Check{y}}_{n,p,u} \mathbf{\Check{y}}_{n,p,u}^H\big\}\Big)^{-1} \mathbf{\Check{y}}_{n,p,u} \nonumber \\
&= \frac{\sqrt{\tau p_{\text{ul}}}~ \beta_{n,u}}{\tau p_{\text{ul}}~ \beta_{n,u} + 1} \mathbf{\Check{y}}_{n,p,u} .
\end{align}
Then, the variance of the estimated channel $\mathbf {\hat{c}}_{n,u}$ is
\begin{equation} \label{eq:gamanu}
\mathbf\gamma_{n,u} = \mathbb{E} \big\{|[\mathbf {\hat{c}}_{n,u}]_l |^2\big\}  = \frac{\tau p_{\text{ul}}~ \beta_{n,u}^2}{\tau p_{\text{ul}}~ \beta_{n,u} + 1}.
\end{equation}
The estimation error of $\mathbf c_{n,u}$ is $\mathbf{\Tilde{c}}_{n,u} =\mathbf{c}_{n,u}-\mathbf{\hat{c}}_{n,u}\sim \mathcal{CN}(\mathbf0,(\beta_{n,u}-\gamma_{n,u}) \mathbf I_{L})$. 

Similarly, the MMSE estimate of the $k_m$-th multicast user $\mathbf t_{n,m,k}$ is given by 
\begin{align}
     \label{eq:tnmkh}
\mathbf {\hat{t}}_{n,m,k} & = \mathbb{E}\big\{\mathbf t_{n,m,k} \mathbf{\Check{y}}_{n,p,m}^H\big\} \Big(\mathbb{E} \big\{\mathbf{\Check{y}}_{n,p,m} \mathbf{\Check{y}}_{n,p,m}^H\big\}\Big)^{-1} \mathbf{\Check{y}}_{n,p,m} \nonumber \\
&= \Big(\sqrt{\tau p_{\text{ul}}} \mathbb{E}\big\{ \mathbf  t_{n,m,k} \mathbf t_{n,m,k}^H\big\} + \mathbb{E}\big\{\mathbf t_{n,m,k} (\boldsymbol{\psi}_{n,p}'')^H\big\}\Big) \nonumber\\
&\hspace{-0.9 em}\times\!\Big(\tau p_{\text{ul}} \mathbb{E} \big\{\!\!\! \sum_{t \in \mathcal{K}_m}\!\!\! \mathbf t_{n,m,t} \mathbf t_{n,m,t}^H  \big\}\! +\! \mathbb{E} \big\{ \boldsymbol{\psi}_{n,p}''(\boldsymbol{\psi}_{n,p}'')^H\big\}\!\Big)^{-1}\!\!\!\!\mathbf{\Check{y}}_{n,p,m} \nonumber \\
&= \frac{\sqrt{\tau p_{\text{ul}}}~ \bar{\beta}_{n,m,k}}{\tau p_{\text{ul}}~ \sum_{t \in \mathcal{K}_m} \bar{\beta}_{n,m,t} + 1} \mathbf{\Check{y}}_{n,p,m},
\end{align}
where $\mathbf{\Check{y}}_{n,p,m}$ is obtained by projecting the received signal \eqref{eq:ynp} on the pilot sequence $\boldsymbol{\varphi}_{m}$ assigned to the $m$-th multicast group, while $\boldsymbol{\psi}_{n,p}''=\mathbf \Psi_{n,p} \boldsymbol{\varphi}_{m}\sim \mathcal{CN} (\mathbf 0,\mathbf I_{L})$. Thus,
\begin{equation} \label{eq:ynpm}
    \mathbf{\Check{y}}_{n,p,m} = \mathbf Y_{n,p} \boldsymbol{\varphi}_{m} = \sqrt{\tau p_{\text{ul}}}  \sum\nolimits_{k \in \mathcal{K}_m} \mathbf t_{n,m,k}  + \boldsymbol{\psi}_{n,p}''.
\end{equation} 
The estimation error of $\mathbf t_{n,m,k}$ is $\mathbf{\Tilde{t}}_{n,m,k} =\mathbf{t}_{n,m,k}-\mathbf{\hat{t}}_{n,m,k} \sim \mathcal{CN}(\mathbf0,(\bar{\beta}_{n,m,k}-\bar{\mathbf\gamma}_{n,m,k}) \mathbf I_{L})$ and the variance of $\mathbf{\hat{t}}_{n,m,k} $ is
\begin{equation} \label{eq:gamanmk}
\bar{\mathbf\gamma}_{n,m,k} = \mathbb{E} \big\{|[\mathbf {\hat{t}}_{n,m,k}]_l |^2\big\} = \frac{\tau p_{\text{ul}}~ \bar{\beta}_{n,m,k}^2}{\tau p_{\text{ul}}~ \sum_{t \in \mathcal{K}_m} \bar{\beta}_{n,m,t} + 1}.
\end{equation}
By employing the co-pilot assignment strategy presented in \cite{yang:SPAWC:2013}, we obtain
\begin{align} \label{eq:tnmh}
\mathbf {\hat{t}}_{n,m} &= \sum\nolimits_{k \in \mathcal{K}_m}\mathbf{\hat{t}}_{n,m,k}
\nonumber\\
&= \frac{\sqrt{\tau p_{\text{ul}}}~ \sum_{k \in \mathcal{K}_m} \bar{\beta}_{n,m,k}}{\tau p_{\text{ul}}~ \sum_{k \in \mathcal{K}_m} \bar{\beta}_{n,m,k} + 1} \mathbf{\Check{y}}_{n,p,m},
\end{align}
which can be regarded as the channel estimate of the $m$-th multicast group. The mean square of $\mathbf {\hat{t}}_{n,m}$ is 
    \begin{equation} \label{eq:zetanm}
\mathbf\zeta_{n,m} = \mathbb{E} \Big\{\big|[\mathbf {\hat{t}}_{n,m}]_l \big|^2 \Big\} = \frac{\Big(\sqrt{\tau p_{\text{ul}}}~ \sum_{k \in \mathcal{K}_m}\bar{\beta}_{n,m,k}\Big)^2}{\tau p_{\text{ul}}~ \sum_{k \in \mathcal{K}_m} \bar{\beta}_{n,m,t} + 1}.
\end{equation}
\textit{\textbf{Remark 1:} Here, we assume Rayleigh fading channels, and hence, the MMSE channel estimation error follows a complex Gaussian distribution. This choice is reasonable in practical rich-scattering scenarios, and has been extensively used in the CF-mMIMO literature \cite{Ngo:IEEE_G_Net:2018,Hien:JWCOM:2017,Emil:TWC:2020}. However, in some practical scenarios,  the channels may not be Rayleigh faded, and hence, channel estimation errors may not follow an exact or known distribution. Factors, such as propagation environments, hardware impairments, and user mobility will introduce more complex forms of uncertainty. In this context, distributionally
robust optimization techniques—which optimize for the worst-case performance over a set
of possible error distributions—have recently gained attention in wireless communications
for their ability to handle uncertainties without relying on a specific error model \cite{He:TWC:2024}. Similarly, machine learning (ML)-based approaches, particularly data-driven CSI predictors or
uncertainty estimators, can capture environment-dependent patterns in channel variations
and provide more adaptive resource allocation strategies \cite{Wang:TSP:2025,Ge:JSYST:2022}.}
\subsection{Downlink Payload Data Transmission}
For the downlink transmission, we employ both MR and ZF precoding schemes. MR precoding is favored for its low computational complexity, relying only on local channel knowledge. This makes MR particularly suitable for large-scale CF-mMIMO deployments with limited per-AP processing capabilities \cite{Hien:JWCOM:2017,Ngo:IEEE_G_Net:2018,Li:TVT:2022}. ZF precoding, in contrast, requires higher computational complexity but mitigates interference more effectively, thereby significantly enhancing the SE. When the system has a surplus of antennas relative to the number of users, ZF can approach near-optimal performance \cite{Sadeghi:TWC:2018,Emil:TWC:2020}.
Nonetheless, ZF is feasible only when the total number of antennas
exceeds the total number of served unicast users and multicast groups. By considering both MR and ZF precoding, our study captures a meaningful range of trade-offs between complexity and performance.
We begin by defining the binary variables $a_{n,u}$ and $\bar{a}_{n,m}$ to indicate the AP–user associations for the $u$-th unicast user and the $m$-th multicast group, respectively. In particular, we define
\begin{subequations}
\begin{align}
\label{eq:anu}
a_{n,u} &\triangleq 
\begin{cases}
  1, &  \text{if unicast user $u$ is served by AP $n$,}\\
  0, & \mbox{othewise},
\end{cases} 
\\
\bar{a}_{n,m}   &\triangleq  
\begin{cases}
  \!1, &  \text{if multicast group $m$ is served  by  AP $n$,}\\
  \!0, &  \mbox{othewise}.
\end{cases} \label{eq:anm}
  \end{align}     
\end{subequations}
Let $q_u$ and $\bar{q}_m$ denote the symbols allocated to the $u$-th unicast user and the $m$-th multicast group users, respectively, with the conditions $\mathbb{E}\{|q_u|^2\}=\mathbb{E}\{|\bar{q}_m|^2\}=1$ and $\mathbb{E}\{q_u\}=\mathbb{E}\{\bar{q}_m\}=0$. Then, the transmitted signal from the $n$-th AP is
\begin{align}\label{eq:xn}
\boldsymbol{x}_{n}^{\varsigma} =&\sqrt{p_{\text{dl}}} \sum_{u \in \mathcal{U}} a_{n,u} \sqrt{\eta_{n,u}} \boldsymbol{b}_{n,u}^{\varsigma} q_u \nonumber\\
&+ \sqrt{p_{\text{dl}}} \sum_{m \in \mathcal{M}} \bar{a}_{n,m} \sqrt{\bar{\eta}_{n,m}} \bar{\boldsymbol{b}}_{n,m}^{\varsigma} \bar{q}_m,
\end{align}
where $\varsigma\in \{\MR,\ZF\}$ indicates the precoding scheme, i.e., $\varsigma=\MR$ implies the MR precoding and $\varsigma=\ZF$ corresponds to the ZF scheme, $p_{\text{dl}}$ is the maximum normalized transmit power at each AP, while $\eta_{n,u}$ and $\bar{\eta}_{n,m}$ are the power control coefficients allocated to the $u$-th unicast user and the $m$-th multicast group, respectively. Moreover, $\boldsymbol{b}_{n,u}^{\varsigma}$ and $\bar{\boldsymbol{b}}_{n,m}^{\varsigma}$ are the $u$-th unicast user and $m$-th multicast group precoding vectors respectively, given by 
\begin{equation}\label{eq:bnu}
\boldsymbol{b}_{n,u}^{\varsigma} \triangleq
\begin{cases}
  \mathbf{\hat{c}}_{n,u}, & \mbox{if}~~\varsigma = \MR ,\\
  \mathbf{\hat{G}}_n(\mathbf{\hat{G}}^H_n\mathbf{\hat{G}}_n)^{-1} \mathbf{e}_{u}, & \mbox{if}~~\varsigma = \ZF,
\end{cases} 
  \end{equation}
    and
\begin{align}\label{eq:bnm}
\bar{\boldsymbol{b}}_{n,m}^{\varsigma} \triangleq
\begin{cases}
  \mathbf{\hat{t}}_{n,m}, & \mbox{if}~~\varsigma = \MR,\\
  \mathbf{\hat{G}}_n(\mathbf{\hat{G}}^H_n\mathbf{\hat{G}}_n)^{-1} \bar{\mathbf{e}}_{m}, & \mbox{if}~~\varsigma = \ZF,
\end{cases} 
\end{align}
where $\mathbf{\hat{G}}_n=[\mathbf{\hat{c}}_{n,1},\dots,\mathbf{\hat{c}}_{n,U},\mathbf{\hat{t}}_{n,1},\dots,\mathbf{\hat{t}}_{n,M}]_{L \times (U+M)}$, $\mathbf{e}_{u}$ is the $u$-th column of the identity matrix $\mathbf{I}_{U+M}$, and $\bar{\mathbf{e}}_{m}$ is the $m$-th column of $\mathbf{I}_{U+M}$.
The power control coefficients $\eta_{n,u}$ and $\bar{\eta}_{n,m}$ are chosen to satisfy the power constraint 
\begin{align}\label{eq:xnp}
\mathbb{E}\big\{\|\boldsymbol{x}_{n}^{\varsigma}\|^2\big\} =& p_{\text{dl}}\sum_{u \in \mathcal{U}} a_{n,u} \eta_{n,u} \mathbb{E}\big\{\|\boldsymbol{b}_{n,u}^{\varsigma} \|^2\big\} \nonumber \\
& \!+p_{\text{dl}}\sum_{m \in \mathcal{M}} \bar{a}_{n,m} \bar{\eta}_{n,m} \mathbb{E}\big\{\| \bar{\boldsymbol{b}}_{n,m}^{\varsigma} \|^2\big\} \leq p_{\text{dl}}.
\end{align}
For the MR precoding, we have $\mathbb{E}\big\{\|\boldsymbol{b}_{n,u}^\MR\|^2\big\}= L ~\mathbf\gamma_{n,u} $ and $\mathbb{E}\big\{\|\bar{\boldsymbol{b}}_{n,m}^\MR \|^2\big\} = L~\mathbf\zeta_{n,m}$.
For ZF precoding, $\mathbb{E}\big\{\|\boldsymbol{b}_{n,u}^\ZF\|^2\big\} =\frac{1}{(L-U-M)\mathbf\gamma_{n,u}}$ and $\mathbb{E}\big\{\|\bar{\boldsymbol{b}}_{n,m}^\ZF\|^2\big\}=\frac{1}{(L-U-M)\mathbf\zeta_{n,m}}$. Therefore, the power constraint in \eqref{eq:xnp} is expressed as
\begin{align}\label{eq:xnvp}
\begin{cases}
    \!L\big(\!\!\sum\limits_{u \in \mathcal{U}} \!\!a_{n,u} \eta_{n,u} \mathbf\gamma_{n,u}\! + \!\!\!\! \sum\limits_{m \in \mathcal{M}}\!\!\! \bar{a}_{n,m} \bar{\eta}_{n,m} \mathbf\zeta_{n,m}\big) \leq 1, & \!\!\!\!\mbox{if}~~\varsigma \!=\! \MR,
    \\
    \!\frac{1}{(L-U-M)}\big(\!\!\sum\limits_{u \in \mathcal{U}} \!a_{n,u}  \frac{\eta_{n,u}}{\mathbf\gamma_{n,u}} + \!\!\!\sum\limits_{m \in \mathcal{M}} \!\!\bar{a}_{n,m}  \frac{\bar{\eta}_{n,m}}{\mathbf\zeta_{n,m}} \big)\leq1, & \!\!\!\!\mbox{if}~~\varsigma\! = \!\ZF.
\end{cases}
\end{align}
Then, the received signal at the $u$-th unicast user is given by
\begin{align}
    \label{eq:rumrb}
r_{u}^{\varsigma}&=\sqrt{p_{\text{dl}}} \sum_{n \in \mathcal{N}} a_{n,u} \sqrt{\eta_{n,u}} \mathbf{c}_{n,u}^H \boldsymbol{b}_{n,u}^\varsigma q_u\nonumber\\
&+ \sqrt{p_{\text{dl}}} \sum_{u' \in \mathcal {U},u' \neq u} \sum_{n \in \mathcal{N}} a_{n,u'} \sqrt{\eta_{n,u'}} \mathbf{c}_{n,u}^H \boldsymbol{b}_{n,u'}^\varsigma q_{u'} \nonumber\\
&+ \sqrt{p_{\text{dl}}} \sum_{m \in \mathcal{M}} \sum_{n \in \mathcal{N}} \bar{a}_{n,m} \sqrt{\bar{\eta}_{n,m}} \mathbf{c}_{n,u}^H \bar{\boldsymbol{b}}_{n,m}^\varsigma \bar{q}_m +{\psi}_{u},
\end{align}
where ${\psi}_{u}$ is the additive Gaussian noise $\mathcal{CN} (0,1)$ at the $u$-th unicast user. The received signal at the $k_m$-th multicast user can be written as
\begin{align}
\label{eq:rmkmrb}
r_{m,k}^{\varsigma}&=\sqrt{p_{\text{dl}}} \sum_{n \in \mathcal{N}} \bar{a}_{n,m} \sqrt{\bar{\eta}_{n,m}} \mathbf{t}_{n,m,k}^H \bar{\boldsymbol{b}}_{n,m}^\varsigma \bar{q}_m \nonumber\\
&+ \sqrt{p_{\text{dl}}} \sum_{m' \in \mathcal{M},m' \neq m} \sum_{n \in \mathcal{N}} \bar{a}_{n,m'} \sqrt{\bar{\eta}_{n,m'}} \mathbf{t}_{n,m,k}^H \bar{\boldsymbol{b}}_{n,m'}^\varsigma \bar{q}_{m'}\nonumber\\
&+ \sqrt{p_{\text{dl}}} \sum_{u \in \mathcal{U}} \sum_{n \in \mathcal{N}} a_{n,u} \sqrt{\eta_{n,{u}}} \mathbf{t}_{n,m,k}^H \boldsymbol{b}_{n,u}^\varsigma q_{{u}} + \psi_{m,k},
\end{align}
where $\psi_{m,k}$ is the additive Gaussian noise $\mathcal{CN} (0,1)$ at the $k_m$-th multicast user.

\section{Achievable Spectral Efficiency}\label{SEs}
In this section, we derive the SE achieved by the $u$-th unicast user and  the   $k_m$-th multicast user under   ZF and MR precoding schemes,  using the use-and-then-forget bounding technique \cite{Du:TCOM:2021}. In general, the achievable SE (bit/s/Hz) is  calculated as 
$$\mathrm{SE}_u^{\varsigma}  = \frac{T-\tau}{T} \log_{2}\big({1+\mathrm{SINR}_u^{\varsigma}} \big), ~ \text{ for unicast users},$$  
and 
$$\mathrm{SE}_{m,k}^{\varsigma}  = \frac{T-\tau}{T} \log_{2}\big({1+\mathrm{SINR}_{m,k}^{\varsigma}} \big), ~ \text{for multicast users},$$   
where $\mathrm{SINR}_u^{\varsigma}$ and $\mathrm{SINR}_{m,k}^{\varsigma}$ are given, respectively, as
\begin{equation} \label{eq:sinru}
\begin{split}
&\mathrm{SINR}_u^{\varsigma} = \\
&\frac{|{\text{DS}}_u^{\varsigma}|^2}{\mathbb{E} \{|\text{BU}_u^{\varsigma} |^2\} +\!\!\!\!\sum\limits_{u' \in \mathcal {U},u' \neq u}\!\!\!\! \mathbb{E} \{| \text{UI}_{u,u'}^{\varsigma}|^2\} + \!\!\!\!\sum\limits_{m \in \mathcal{M}}\! \mathbb{E} \{|\text{MI}_{u,m}^{\varsigma}|^2\} +1},
\end{split}
\end{equation}
and
\begin{equation} \label{eq:sinrmk}
\begin{split}
&\mathrm{SINR}_{m,k}^{\varsigma} = \\
&\frac{| {\text{DS}}_{m,k}^{\varsigma}|^2}{\mathbb{E} \{|\text{BU}_{m,k}^{\varsigma}|^2\}\! + \!\!\!\!\!\!\sum\limits_{m' \in \mathcal{M},m' \neq m}\!\!\!\!\!\!\!\!\!\! \mathbb{E} \{|\text{MI}_{m,k,m'}^{\varsigma} |^2\} \!+\!\!\!\sum\limits_{u \in \mathcal {U}}\!\! \mathbb{E} \{| \text{UI}_{m,k,u}^{\varsigma} |^2\} \!+\!\!1},
\end{split}
\end{equation}
where $\text{DS}_u^{\varsigma}$ and ${\text{DS}}_{m,k}^{\varsigma}$ are the desired signals, while $\text{BU}_u^{\varsigma}$ and $\text{BU}_{m,k}^{\varsigma}$ are the beamforming uncertainties for the unicast user and multicast user, respectively. Additionally, $\text{UI}_{u,u'}^{\varsigma}$, $\text{UI}_{m,k,u}^{\varsigma}$ and  $\text{MI}_{m,k,m'}^{\varsigma}$, $\text{MI}_{u,m}^{\varsigma}$  are the interfering signals from other unicast users and multicast groups, respectively.
\begin{figure*}
\begin{align}
\mathrm{SINR}_{u}^{\MR} &= \frac{\big( \sqrt{p_{\text{dl}}}~L~\sum\nolimits_{n \in \mathcal{N}} a_{n,u} \sqrt{\eta_{n,u}} \mathbf\gamma_{n,u}\big)^2}{ p_{\text{dl}}\big( L\sum\nolimits_{u' \in \mathcal {U}} \sum\nolimits_{n \in \mathcal{N}} a_{n,u'} \eta_{n,u'} \beta_{n,u} \gamma_{n,u'} +  L\sum\nolimits_{m \in \mathcal{M}} \sum\nolimits_{n \in \mathcal{N}} \bar{a}_{n,m} \bar{\eta}_{n,m}   \beta_{n,u}  \zeta_{n,m} \big)+1}\label{eq:sinrumra} ~\tag{23} \\[0.5em]
\mathrm{SINR}_u^{\ZF} &= \frac{\big( \sqrt{p_{\text{dl}}}\sum\nolimits_{n \in \mathcal{N}} a_{n,u} \sqrt{\eta_{n,u}}\big)^2}{p_{\text{dl}} \Big(  \sum\nolimits_{u' \in \mathcal{U}} \sum\nolimits_{n \in \mathcal{N}} a_{n,u'} \eta_{n,u'}  \frac{(\beta_{n,u}-\gamma_{n,u})}{(L-U-M)\mathbf\gamma_{n,u'}} +  \sum\nolimits_{m \in \mathcal{M}}  \sum\nolimits_{n \in \mathcal{N}} \bar{a}_{n,m} \bar{\eta}_{n,m}  \frac{(\beta_{n,u}-\gamma_{n,u})}{(L-U-M)\mathbf\zeta_{n,m}}\Big) +1}\label{eq:sinruzfa}  ~\tag{24}\\[0.5em]
\mathrm{SINR}_{m,k}^{\MR} &= \frac{( \sqrt{p_{\text{dl}}} ~ L\sum\nolimits_{n \in \mathcal{N}} \bar{a}_{n,m} \sqrt{\bar{\eta}_{n,m}} ~ \sqrt{\zeta_{n,m}\bar{\mathbf\gamma}_{n,m,k}} )^2}{ p_{\text{dl}} \big(L \sum\nolimits_{m' \in \mathcal{M}} \sum\nolimits_{n \in \mathcal{N}} \bar{a}_{n,m'} \bar{\eta}_{n,m'}\bar{\beta}_{n,m,k}\zeta_{n,m'} + L \sum\nolimits_{u \in \mathcal {U}} \sum\nolimits_{n \in \mathcal{N}} a_{n,u} \eta_{n,u} \bar{\beta}_{n,m,k} \gamma_{n,u}\big) +1}\label{eq:sinrmkmra} ~\tag{26}\\[0.5em]
\mathrm{SINR}_{m,k}^{\ZF} &= \frac{\big(\sqrt{p_{\text{dl}}}\sum\nolimits_{n \in \mathcal{N}} \bar{a}_{n,m} \sqrt{\bar{\eta}_{n,m}}\big)^2}{p_{\text{dl}} \Big(\sum\nolimits_{m' \in \mathcal{M}}\sum\nolimits_{n \in \mathcal{N}} \bar{a}_{n,m'} \bar{\eta}_{n,m'}  \frac{\bar{\beta}_{n,m,k}-\bar{\mathbf\gamma}_{n,m,k}}{(L-U-M)\mathbf\zeta_{n,m'}}+\sum\nolimits_{u \in \mathcal{U}}\sum\nolimits_{n \in \mathcal{N}} a_{n,u} \eta_{n,u}  \frac{\bar{\beta}_{n,m,k}-\bar{\mathbf\gamma}_{n,m,k}}{(L-U-M)\mathbf\gamma_{n,u}}\Big)+1}\label{eq:sinrmkzfa} ~\tag{27}
\end{align} 
\hrulefill
\end{figure*}
\subsection{SINR at the $u$-th Unicast User}\label{sinrmr}
Using~\eqref{eq:rumrb}, the corresponding SINR terms for the $u$-th unicast user can be written as follows:
\begin{align}\label{eq:dsuv} 
&\text{DS}_u^{\varsigma} = \sqrt{p_{\text{dl}}}~ \mathbb{E} \left\{\sum\nolimits_{n \in \mathcal{N}} a_{n,u} \sqrt{\eta_{n,u}} \mathbf{c}_{n,u}^H \boldsymbol{b}_{n,u}^\varsigma \right\}, 
\nonumber\\
&\text{BU}_u^{\varsigma} = \sqrt{p_{\text{dl}}} \sum\nolimits_{n \in \mathcal{N}} a_{n,u} \sqrt{\eta_{n,u}} \mathbf{c}_{n,u}^H \boldsymbol{b}_{n,u}^\varsigma \nonumber \\ 
& \hspace{6em} - \mathbb{E} \left\{ \sqrt{p_{\text{dl}}} \sum\nolimits_{n \in \mathcal{N}} a_{n,u} \sqrt{\eta_{n,u}} \mathbf{c}_{n,u}^H \boldsymbol{b}_{n,u}^\varsigma \right\} ,\nonumber
\\
&\text{UI}_{u,u'}^{\varsigma} = \sqrt{p_{\text{dl}}} \sum\nolimits_{n \in \mathcal{N}} a_{n,u'} \sqrt{\eta_{n,u'}} \mathbf{c}_{n,u}^H \boldsymbol{b}_{n,u'}^\varsigma,\nonumber
\\
&\text{MI}_{u,m}^{\varsigma} =  \sqrt{p_{\text{dl}}}   \sum\nolimits_{n \in \mathcal{N}} \bar{a}_{n,m} \sqrt{\bar{\eta}_{n,m}} \mathbf{c}_{n,u}^H \bar{\boldsymbol{b}}_{n,m}^\varsigma.
\end{align}
\begin{proposition}
The closed-form expressions for the received SINR at the $u$-th unicast user under MR and ZF precoding, $\mathrm{SINR}_u^{\MR}$ and $\mathrm{SINR}_{u}^{\ZF}$, are given by \eqref{eq:sinrumra} and \eqref{eq:sinruzfa}, respectively, shown at the top of the next page.

\end{proposition}
\begin{proof}
The proof is provided in Appendix \ref{proumr}. 
\end{proof}

\subsection{SINR at the $k_m$-th Multicast User}\label{sinrzf}
From \eqref{eq:rmkmrb}, the corresponding SINR terms for    the $k_m$-th multicast user can be written as follows:
\setcounter{equation}{24}
\begin{align}\label{eq:dsmkv}
&\text{DS}_{m,k}^{\varsigma} = \sqrt{p_{\text{dl}}} ~\mathbb{E} \left\{\sum\nolimits_{n \in \mathcal{N}} \bar{a}_{n,m} \sqrt{\bar{\eta}_{n,m}} \mathbf{t}_{n,m,k}^H \bar{\boldsymbol{b}}_{n,m}^\varsigma \right\},\nonumber \\
&\text{BU}_{m,k}^{\varsigma} = \sqrt{p_{\text{dl}}} \sum_{n \in \mathcal{N}} \bar{a}_{n,m} \sqrt{\bar{\eta}_{n,m}} \mathbf{t}_{n,m,k}^H \bar{\boldsymbol{b}}_{n,m}^\varsigma \nonumber \\
& \hspace{5em}- \mathbb{E} \left\{ \sqrt{p_{\text{dl}}}\sum\nolimits_{n \in \mathcal{N}} \bar{a}_{n,m} \sqrt{\bar{\eta}_{n,m}} \mathbf{t}_{n,m,k}^H \bar{\boldsymbol{b}}_{n,m}^\varsigma \right\},\nonumber \\
&\text{MI}_{m,k,m'}^{\varsigma} =  \sqrt{p_{\text{dl}}} \sum\nolimits_{n \in \mathcal{N}} \bar{a}_{n,m'} \sqrt{\bar{\eta}_{n,m'}} \mathbf{t}_{n,m,k}^H \bar{\boldsymbol{b}}_{n,m'}^\varsigma ,\nonumber \\
&\text{UI}_{m,k,u}^{\varsigma} =  \sqrt{p_{\text{dl}}}  \sum\nolimits_{n \in \mathcal{N}} a_{n,u} \sqrt{\eta_{n,{u}}} \mathbf{t}_{n,m,k}^H \boldsymbol{b}_{n,u}^\varsigma.
\end{align}
\begin{proposition}
The closed-form expressions for the received SINRs at the $k_m$-th multicast user with MR precoding design, $\mathrm{SINR}_{m,k}^{\MR}$, and with ZF precoding design, $\mathrm{SINR}_{m,k}^{\ZF}$, are given by \eqref{eq:sinrmkmra} and \eqref{eq:sinrmkzfa}, respectively, shown at the top of the next page.    
\end{proposition}
\begin{proof}
The proof is provided in Appendix  \ref{promkmr}. 
\end{proof}

\section{Joint AP Selection and Power Allocation Optimization}
Here, we aim to jointly optimize the power allocation coefficients $\boldsymbol{\eta} \triangleq \{\eta_{n,u},\bar{\eta}_{n,m}\}$ and user association $\boldsymbol{a} \triangleq \{a_{n,u},\bar{a}_{n,m}\}$ in order to maximize the weighted SSE of unicast and multicast users, subject to the SE requirements, fronthaul constraints, and the per-AP transmit power constraint in \eqref{eq:xnvp}.\footnote{To explicitly address fairness, the optimization problem can be extended by incorporating fairness-driven criteria. For example, a max-min fairness approach can be adopted
to maximize the minimum SE among multicast groups, ensuring a more uniform performance distribution.} To facilitate our algorithmic design and present a general formulation of the problem,  we introduce the following notation:
 \begin{itemize}
     \item $\rho^{\varsigma} \triangleq
\begin{cases}
  1/L, & \varsigma = \MR,
  \\
  L-U-M, & \varsigma = \ZF.
  \end{cases} $

  \item $\boldsymbol{\theta}^{\varsigma}\triangleq \big[(\boldsymbol{\theta}_1^{\varsigma})^T,\ldots,(\boldsymbol{\theta}_N^{\varsigma})^T\big]^T$, where $\boldsymbol{\theta}_n^{\varsigma}=[\theta_{n,1}^{\varsigma},\dots,\theta_{n,U}^{\varsigma},\bar{\theta}_{n,1}^{\varsigma},\dots,\bar{\theta}_{n,M}^{\varsigma}]^T$, and 
\begin{align}
\theta_{n,u}^{\varsigma} &\triangleq
\begin{cases}
  \sqrt{\eta_{n,u}\gamma_{n,u}}, & \varsigma = \MR,
  \\
  \frac{\sqrt{\eta_{n,u}}}{\sqrt{\gamma_{n,u}}}, & \varsigma = \ZF,
\end{cases} \nonumber\\
\bar{\theta}_{n,m}^{\varsigma} &\triangleq
\begin{cases}
 \sqrt{\bar{\eta}_{n,m}\zeta_{n,m}}, & \varsigma = \MR,
 \\
\frac{\sqrt{\bar{\eta}_{n,m}}}{\sqrt{\zeta_{n,m}}}, & \varsigma = \ZF.
\end{cases}\nonumber
\end{align}

\item $\Lambda_{n,u}^{\varsigma} \triangleq
    \begin{cases}
       L~\sqrt{\mathbf\gamma_{n,u}} & \mbox{${\varsigma} = \MR$},
       \\
        \sqrt{\mathbf\gamma_{n,u}} & \mbox{${\varsigma} = \ZF$}.
    \end{cases}$
\item $\Theta_{n,u}^{\varsigma} \triangleq
    \begin{cases}
       L~\beta_{n,u} & \mbox{${\varsigma} = \MR$},
       \\
        \frac{\beta_{n,u}-\gamma_{n,u}}{(L-U-M)} & \mbox{${\varsigma}= \ZF$}.
    \end{cases}$

    \item $\bar{\Lambda}_{n,m,k}^{\varsigma} \triangleq
    \begin{cases}
       L~\sqrt{\bar{\mathbf\gamma}_{n,m,k}}
       & \mbox{${\varsigma} = \MR$},
       \\
        \sqrt{\zeta_{n,m}} & \mbox{${\varsigma} = \ZF$}.
    \end{cases}$
    
\item $\bar{\Theta}_{n,m,k}^{\varsigma} \triangleq
    \begin{cases}
       L~\bar{\beta}_{n,m,k} & \mbox{${\varsigma} = \MR$},
       \\
        \frac{\bar{\beta}_{n,m,k}-\bar{\gamma}_{n,m,k}}{(L-U-M)} & \mbox{${\varsigma} = \ZF$}.
    \end{cases}$
 \end{itemize}
\setcounter{equation}{27}

 Moreover, by considering~\eqref{eq:anu} and~\eqref{eq:anm},  we enforce
\begin{align} \label{eq:thataa}
   \theta_{n,u}^\varsigma &= 0 , \mathrm{if}~a_{n,u}= 0 ~\forall  ~ n,u,  \nonumber\\ \bar{\theta}_{n,m}^\varsigma &= 0 , \mathrm{if}~\bar{a}_{n,m} = 0 ~\forall  ~ n,m,
\end{align}
to ensure that if AP $n$ does not associate with unicast  user $u$ (multicast user $k_m$), the transmit power ${  p_{\text{dl}} (\theta_{n,u}^\varsigma)^2}/{\gamma_{n,u}}$  towards  unicast  user $u$ (${  p_{\text{dl}} (\bar{\theta}_{n,m}^\varsigma)^2}/{\zeta_{n,m}}$   towards  multicast user $k_m$) is zero.

Now, we can rewrite the received SINR at unicast and multicast users as \eqref{eq:sinruth} and \eqref{eq:sinrmkth}, respectively, shown at the top of the next page. We highlight that the user association  $\boldsymbol{a}$ only affects the SE expressions  via parameter $\boldsymbol{\theta}^{\varsigma}$  and~\eqref{eq:thataa}. Accordingly, the optimization problem is formulated as

\setcounter{equation}{30}
\begin{subequations}\label{eq:optpg}
\begin{align}
&\min\limits_{\boldsymbol{a,\theta} } -\Big(\!w_1 \!\sum_{u \in \mathcal {U}}\mathrm{SE}_{u}^{\varsigma} (\boldsymbol{\theta}^{\varsigma}) + w_2 \!\sum_{m \in \mathcal{M}} \sum_{k \in \mathcal{K}_m} \!\mathrm{SE}_{m,k}^{\varsigma} (\boldsymbol{\theta}^{\varsigma}) \Big), \\
&\text{s.t.:} ~
\nonumber\\
&C_1: \mathrm{SE}_{u}^{\varsigma} (\boldsymbol{\theta}^{\varsigma}) \ge \mathrm{SE}_{QoS} , \mathrm{SE}_{m,k}^{\varsigma} (\boldsymbol{\theta}^{\varsigma}) \ge \bar{\mathrm{SE}}_{QoS}, \forall u, m, \label{eq:cons1g}\\
&C_2: \theta_{n,u}^{\varsigma}  \geq 0 ~,~  \bar{\theta}_{n,m}^{\varsigma} \geq 0, ~~ \forall ~n,u,m, \label{eq:cons2g} \\
&C_3:  \sum\nolimits_{u \in \mathcal {U}} (\theta_{n,u}^{\varsigma})^2   +  \sum\nolimits_{m \in \mathcal{M}} (\bar{\theta}_{n,m}^{\varsigma})^2  \leq \rho^{\varsigma},  ~ \forall ~n, \label{eq:cons3g}
\\
& C_4: \sum_{u \in \mathcal {U}} a_{n,u} \mathrm{SE}_{u}^{\varsigma}(\boldsymbol{\theta}^{\varsigma}) \nonumber\\
&~~~~~~+ \!\!\!\sum_{m \in \mathcal{M}}\bar{a}_{n,m}\sum_{k \in \mathcal{K}_m}  \mathrm{SE}_{m,k}^{\varsigma}(\boldsymbol{\theta}^{\varsigma}) \leq C_{\mathrm{max},n},~~ \forall ~ n, \label{eq:cons4g}
\\
&C_5: \sum\nolimits_{n \in \mathcal{N}} a_{n,u} \ge 1,~ \sum\nolimits_{n \in \mathcal{N}} \bar{a}_{n,m} \ge 1,~~ \forall ~ u,m, \label{eq:cons5g} \\
&C_6: \sum\nolimits_{u \in \mathcal {U}} a_{n,u} +\sum\nolimits_{m \in \mathcal{M}} \bar{a}_{n,m}  \leq K_{\mathrm{max},n},~~ \forall ~ n, \label{eq:cons6g}
\end{align}
\end{subequations}
where $w_1$ and $w_2$,  $w_1+w_2=1$, are the weighting coefficients, while $\mathrm{SE}_{QoS}$ and $\mathrm{\bar{SE}}_{QoS}$ in~\eqref{eq:cons1g} denote the minimum SE requirements for the  unicast users and  multicast users, respectively. The constraint \eqref{eq:cons5g} guarantees that at least one AP serves each unicast user and at least one AP serves each multicast group, while constraint~\eqref{eq:cons6g} guarantees that   the maximum number of unicast users and multicast groups served by each AP is $ K_{\mathrm{max},n}, 1 \leq K_{\mathrm{max},n} \leq U+M$. Additionally, the constraint in \eqref{eq:cons4g} ensures that the fronthaul consumption between the CPU and 
 each AP does not exceed the threshold $C_{\mathrm{max},n}$.  It is important to note that under stringent fronthaul constraints, the achievable SE of both unicast and multicast users can be degraded. To mitigate this, our framework incorporates   minimum SE constraints, ensuring that all users meet basic QoS requirements even under limited fronthaul capacity. Beyond this, practical enhancements, such as quantization and signal compression, as well as fronthaul-aware user scheduling, could be integrated to further improve resilience and performance in constrained deployments~\cite{Yasaman:TWC:2023}.
 \\
\begin{figure*}
\begin{align} 
\mathrm{SINR}_{u}^{\varsigma}(\boldsymbol{\theta}^{\varsigma})&\triangleq\frac{U_u^{\varsigma}(\boldsymbol{\theta}^{\varsigma})}{V_u^{\varsigma}(\boldsymbol{\theta}^{\varsigma})}=\frac{\big(  \sqrt{p_{\text{dl}}}\sum\nolimits_{n \in \mathcal{N}}  \theta_{n,u}^{\varsigma} \Lambda_{n,u}^{\varsigma}\big)^2}{ p_{\text{dl}}~ \sum\nolimits_{u' \in \mathcal {U}} \sum\nolimits_{n \in \mathcal{N}}  (\theta_{n,u'}^{\varsigma})^2 \Theta_{n,u}^{\varsigma} + p_{\text{dl}}  \sum\nolimits_{m \in \mathcal{M}} \sum\nolimits_{n \in \mathcal{N}} (\bar{\theta}_{n,m}^{\varsigma})^2   \Theta_{n,u}^{\varsigma} +1}\label{eq:sinruth} ~\tag{29}\\[0.5em]
\mathrm{SINR}_{m,k}^{\varsigma}(\boldsymbol{\theta}^{\varsigma})&\triangleq\frac{U_{m,k}^{\varsigma}(\boldsymbol{\theta}^{\varsigma})}{V_{m,k}^{\varsigma}(\boldsymbol{\theta}^{\varsigma})}=\frac{\big( \sqrt{p_{\text{dl}}} \sum\nolimits_{n \in \mathcal{N}} \bar{\theta}_{n,m}^{\varsigma}  \bar{\Lambda}_{n,m,k}^{\varsigma} \big)^2}{p_{\text{dl}}  \sum\nolimits_{m' \in \mathcal{M}} \sum\nolimits_{n \in \mathcal{N}} (\bar{\theta}_{n,m'}^{\varsigma})^2 \bar{\Theta}_{n,m,k}^{\varsigma} + p_{\text{dl}} \sum\nolimits_{u \in \mathcal {U}} \sum\nolimits_{n \in \mathcal{N}}  (\theta_{n,u}^{\varsigma})^2 \bar{\Theta}_{n,m,k}^{\varsigma} +1} \label{eq:sinrmkth} ~\tag{30}
\end{align}
\hrulefill
\end{figure*}
\subsection{APG-Based Optimization Approach}
The joint optimization problem~\eqref{eq:optpg} is a non-convex mixed-integer problem that is difficult to solve.
First, to address binary constraints  $\eqref{eq:anu}$ and $\eqref{eq:anm}$,  we note that $x \in$ $\{0,1\}$ is equivalent to $x \in[0,1]\quad \& \quad x-x^2 \leq 0$\cite{Vu:IOT:2022}. Thus we  replace $\eqref{eq:anu}$ and $\eqref{eq:anm}$  with the following AP association constraints:
\begin{subequations}
\begin{align}
S_u(\mathbf{a}) &\triangleq \sum_{u \in \mathcal {U}}   \sum_{n \in \mathcal{N}} (a_{n,u}-a_{n,u}^2) \leq 0, ~~ 0 \leq a_{n,u} \leq 1, ~\forall n, u, \label{eq:conua} \\
\bar{S}_{m}(\mathbf{a}) &\triangleq \!\!\! \sum_{m \in \mathcal{M}}   \sum_{n \in \mathcal{N}} (\bar{a}_{n,m}-\bar{a}_{n,m}^2) \leq 0,~
0 \leq \bar{a}_{n,m} \leq 1, \forall n, m, \label{eq:conma}
\end{align}
\end{subequations}
respectively.
Thus,
\begin{equation}\label{eq:tha}
\begin{split}
(\theta_{n,u}^\varsigma)^2 \leq a_{n,u}, ~~~~~
(\bar{\theta}_{n,m}^\varsigma)^2 \leq \bar{a}_{n,m}.
\end{split}
\end{equation}
Next, we define the new parameter $\boldsymbol{z}\triangleq[\boldsymbol{z}^T_1,\ldots,\boldsymbol{z}^T_N]^T$,   where
$\boldsymbol{z}_n=[z_{n,1},\dots,z_{n,U},\bar{z}_{n,1},\dots,\bar{z}_{n,M}]^T$,
while  $z_{n,u}^2 \triangleq a_{n,u}$ and $\bar{z}_{n,m}^2 \triangleq \bar{a}_{n,m}$ with
\begin{equation}\label{eq:znuznm}
\begin{split}
0\leq z_{n,u} \leq 1 ~~\text{and} ~~0\leq \bar{z}_{n,m}\leq 1.
\end{split}
\end{equation}
Therefore, constraint \eqref{eq:cons5g} can be re-expressed as
\begin{align}\label{eq:znuznmkh}
    \sum\nolimits_{u \in \mathcal{U}} z_{n,u}^2 + \sum\nolimits_{m \in \mathcal{M}} \bar{z}_{n,m}^2 \leq K_{\mathrm{max},n},~~ \forall ~ n.
\end{align}
In addition, constraints \eqref{eq:cons3g}, \eqref{eq:conua}, \eqref{eq:conma},  \eqref{eq:cons5g}, \eqref{eq:tha} and \eqref{eq:cons4g} can be respectively replaced by
\begin{align}
C_{1,u}(\boldsymbol{\theta}^{\varsigma}) &\triangleq \sum\nolimits_{u \in \mathcal {U}} \big[\max (0,\mathrm{SE}_{QoS} -\mathrm{SE}_{u}^{\varsigma}(\boldsymbol{\theta}^{\varsigma}))\big]^2\leq 0,\label{eq:c1uc1m}\\
\bar{C}_{1,m}(\boldsymbol{\theta}^{\varsigma}) \!&\triangleq\!\!\! \sum_{m \in \mathcal{M}} \!\sum_{k \in \mathcal{K}_m} \!\!\big[\max (0,\mathrm{\bar{SE}}_{QoS}\!-\!\mathrm{SE}_{m,k}^{\varsigma} (\boldsymbol{\theta}^{\varsigma}))\!\big]^2\! \leq\! 0,
\\
C_{2,u}(\mathbf{z}) &\triangleq \sum\nolimits_{u \in \mathcal {U}} \sum\nolimits_{n \in \mathcal{N}}(z_{n,u}^2-z_{n,u}^4) \leq 0,\\
\bar{C}_{2,m}(\mathbf{z}) &\triangleq \sum\nolimits_{m \in \mathcal{M}} \sum\nolimits_{n \in \mathcal{N}} (\bar{z}_{n,m}^2-\bar{z}_{n,m}^4) \leq 0,\label{eq:c2uc2m}
\end{align}
\begin{align}
C_{3,u}(\boldsymbol{\theta}^{\varsigma}, \mathbf{z}) &\triangleq \sum\nolimits_{u \in \mathcal {U}} \bigg(\Big[\max \big(0,1-\sum\nolimits_{n \in \mathcal{N}} z_{n,u}^2 \big)\Big]^2\nonumber\\
&\hspace{-2em}+\sum\nolimits_{n \in \mathcal{N}} \Big[\max \big(0, (\theta_{n,u}^{\varsigma})^2 - z_{n,u}^2 \big)\Big]^2\bigg)\leq 0,\\
\bar{C}_{3,m}(\boldsymbol{\theta}^{\varsigma}, \mathbf{z}) &\triangleq \sum\nolimits_{m \in \mathcal{M}} \bigg( \Big[\max \big(0,1-\sum_{n \in \mathcal{N}} \bar{z}_{n,m}^2 \big) \Big]^2\nonumber\\
&\hspace{-2em}+\sum\nolimits_{n \in \mathcal{N}} \Big[\max \big(0, (\bar{\theta}_{n,m}^{\varsigma})^2 - \bar{z}_{n,m}^2 \big)\Big]^2\bigg)\leq 0,
\end{align}
\begin{align}
C_{4}(\boldsymbol{\theta}^{\varsigma}, \mathbf{z}) &\triangleq \sum\nolimits_{n \in \mathcal {N}}  \bigg[\max \Big(0,\sum_{u \in \mathcal {U}} z_{n,u}^2 \mathrm{SE}_{u}^{\varsigma}(\boldsymbol{\theta}^{\varsigma}) \nonumber \\
& \hspace{-4em}+\! \sum\nolimits_{m \in \mathcal{M}}\!\!\bar{z}_{n,m}^2\sum\nolimits_{k \in \mathcal{K}_m}  \mathrm{SE}_{m,k}^{\varsigma}(\boldsymbol{\theta}^{\varsigma})\!- \!C_{\mathrm{max},n}\Big)\bigg]^2 \!\!\leq 0. \label{eq:c4um}
\end{align} 

Now, we define
\begin{align}\label{eq:fv}
&g(\boldsymbol{\vartheta}^\varsigma) \triangleq  -\Big(w_1 \sum_{u \in \mathcal {U}} \mathrm{SE}_{u}^{\varsigma}(\boldsymbol{\theta}^\varsigma) + w_2 \sum_{m \in \mathcal{M}} \sum_{k \in \mathcal{K}_m}\mathrm{SE}_{m,k}^{\varsigma}(\boldsymbol{\theta}^\varsigma) \Big) \nonumber \\ 
&+ X \Big[\mu_1 \big(C_{1,u}(\boldsymbol{\theta}^\varsigma)+\bar{C}_{1,m}(\boldsymbol{\theta}^\varsigma)\big)+\mu_2 \big(C_{2,u}(\mathbf{z})+\bar{C}_{2,m}(\mathbf{z})\big) \nonumber \\
&~~~+\mu_3 \big(C_{3,u}(\boldsymbol{\theta}^\varsigma, \mathbf{z})+\bar{C}_{3,m}(\boldsymbol{\theta}^\varsigma, \mathbf{z}) \big) + \mu_4 C_{4}(\boldsymbol{\theta}^\varsigma, \mathbf{z}) \Big],
\end{align}
where $\mu_1, \mu_2, \mu_3, ~\text{and}~\mu_4 $ are positive weights,  $X$ is the Lagrangian multiplier, and $\boldsymbol{\vartheta}^\varsigma \triangleq\big[{(\boldsymbol{\theta}^\varsigma)}^T, \mathbf{z}^T\big]^T$. 
Thus, the optimization problem~\eqref{eq:optpg} can be expressed equivalently as
\begin{equation}\label{eq:apgopt}
\min\limits_{\boldsymbol{\vartheta}^\varsigma \in \widehat{\mathcal{C}}} g(\boldsymbol{\vartheta}^\varsigma),
\end{equation}
where $\widehat{\mathcal{C}} \triangleq \{\eqref{eq:cons2g}, \eqref{eq:cons3g}, \eqref{eq:znuznm},\eqref{eq:znuznmkh}\}$ is the convex feasible set. 
Here, we leverage the APG approach to tackle  joint optimization problem  \eqref{eq:apgopt}. Although the APG approach is suboptimal, it offers significantly lower complexity compared to common SCA algorithms, especially beneficial for handling large-scale CF-mMIMO networks~\cite{Farooq:TCOM:2021,Mai:TWC:2022}.
Specifically, our proposed method to solve problem \eqref{eq:apgopt} is given in \textbf{Algorithm~\ref{algo}}. 
The primary tasks in executing \textbf{Algorithm~\ref{algo}} include computing the gradient of the objective function and performing projections, as outlined below.
\subsubsection{Gradient  of $g(\boldsymbol{\vartheta}^\varsigma)$}
The gradients  $\frac{\partial}{\partial \theta_{n,u}^\varsigma} g(\boldsymbol{\vartheta}^\varsigma)$ and $\frac{\partial}{\partial z_{n,u}} g(\boldsymbol{\vartheta}^\varsigma)$ are given by
\begin{align}
 \frac{\partial}{\partial \theta_{n,u}^\varsigma} g(\boldsymbol{\vartheta}^\varsigma)\! &=\! -w_1 \!\sum\nolimits_{i \in \mathcal {U}} \frac{\partial}{\partial \theta_{n,u}^\varsigma} \mathrm{SE}_{i}^{\varsigma}(\boldsymbol{\theta}^\varsigma) + X \frac{\partial}{\partial \theta_{n,u}^\varsigma} C_{u} (\boldsymbol{\vartheta}^\varsigma),
\\
\frac{\partial}{\partial z_{n,u}} g(\boldsymbol{\vartheta}^\varsigma)\!&=\! -w_1\!\sum\nolimits_{i \in \mathcal {U}} \frac{\partial}{\partial z_{n,u}}\mathrm{SE}_{i}^{\varsigma}(\boldsymbol{\theta}^\varsigma) + X \frac{\partial}{\partial z_{n,u}} C_{u} (\boldsymbol{\vartheta}^\varsigma),
\end{align}
respectively, with $C_u(\boldsymbol{\vartheta}^\varsigma)= \mu_1 C_{1,u}(\boldsymbol{\theta}^\varsigma)+\mu_2 C_{2,u}(\mathbf{z})+\mu_3 C_{3,u}(\boldsymbol{\theta}^\varsigma, \mathbf{z}) + \mu_4 C_{4}(\boldsymbol{\theta}^\varsigma, \mathbf{z})$.
Moreover, $\frac{\partial}{\partial \theta_{n,u}^\varsigma} \mathrm{SE}_{i}^{\varsigma}(\boldsymbol{\theta}^\varsigma)$ is given by
\begin{align}\label{eq:gradthanu}
\frac{\partial}{\partial \theta_{n,u}^\varsigma} &\mathrm{SE}_{i}^{\varsigma}\!(\boldsymbol{\theta}^\varsigma)=\nonumber\\
&\hspace{-1em}\!\!\frac{T-\tau}{T\ln 2}\!\bigg[\!\frac{\frac{\partial}{\partial \theta_{n,u}^\varsigma}(U_i^\varsigma(\boldsymbol{\theta}^\varsigma)+V_i^\varsigma(\boldsymbol{\theta}^\varsigma))}{(U_i^\varsigma(\boldsymbol{\theta}^\varsigma)+V_i^\varsigma(\boldsymbol{\theta}^\varsigma))}-\frac{\frac{\partial}{\partial \theta_{n,u}^\varsigma} V_i^\varsigma(\boldsymbol{\theta}^\varsigma)}{V_i^\varsigma(\boldsymbol{\theta}^\varsigma)}\!\bigg],
\end{align}
with
\begin{align}
&\frac{\partial U_i^\varsigma(\boldsymbol{\theta}^\varsigma)}{\partial \theta_{n,u}^\varsigma} = 
\begin{cases}
2( \sqrt{p_{\text{dl}}}\sum_{n \in \mathcal{N}} \theta_{n,u}^\varsigma \Lambda_{n,u}^\varsigma)(\sqrt{p_{\text{dl}}}\Lambda_{n,u}^\varsigma), & i=u, \\
0, 
& i \neq u, 
\end{cases}\nonumber \\
&\frac{\partial}{\partial \theta_{n,u}^\varsigma} V_i^\varsigma(\boldsymbol{\theta}^\varsigma)= 
\begin{cases}
2~ p_{\text{dl}} \theta_{n,u}^\varsigma \Theta_{n,u}^\varsigma, & i = u , \\
2~ p_{\text{dl}} \theta_{n,u}^\varsigma \Theta_{n,i}^\varsigma, & i \neq u.
\end{cases} 
\end{align}
Moreover,
$ -\sum_{i \in \mathcal {U}} \frac{\partial}{\partial z_{n,u}} \mathrm{SE}_{i}^{\varsigma}(\boldsymbol{\theta}^\varsigma) = 0, ~ \forall~ n, u,i$. In addition,
\begin{align}
\frac{\partial}{\partial \theta_{n,u}} C_{u} (\boldsymbol{\vartheta}^\varsigma) =& - \mu_1 \sum\nolimits_{i \in \mathcal {U}} 2  \max \big(0,\mathrm{SE}_{QoS}-\mathrm{SE}_{i}^{\varsigma}(\boldsymbol{\theta}^\varsigma)\big) \nonumber \\
&\hspace{-4em}\times \frac{\partial}{\partial \theta_{n,u}^\varsigma} \mathrm{SE}_{i}^{\varsigma}(\boldsymbol{\theta}^\varsigma)+ 4 \mu_3  \max \big(0, (\theta_{n,u}^\varsigma)^2 - z_{n,u}^2 \big)\theta_{n,u}^\varsigma \nonumber \\
&\hspace{-4em} +2z_{n,i}^2\mu_4 \bigg[\max \Big(0,\sum\nolimits_{i \in \mathcal {U}} z_{n,i}^2 \mathrm{SE}_{i}^{\varsigma}(\boldsymbol{\theta}^\varsigma) \nonumber \\
& \hspace{-6em} + \!\!\!\!\sum_{m \in \mathcal{M}}\!\!\!\bar{z}_{n,m}^2 \!\!\! \sum_{k \in \mathcal{K}_m} \!\!\! \mathrm{SE}_{m,k}^{\varsigma}(\boldsymbol{\theta}^\varsigma) \! - C_{\mathrm{max},n}\Big)\!\bigg]\!\!  \frac{\partial}{\partial \theta_{n,u}} \mathrm{SE}_{i}^{\varsigma}(\boldsymbol{\theta}^\varsigma),
\end{align}
and
\begin{align}
& \frac{\partial}{\partial z_{n,u}} C_{u} (\boldsymbol{\vartheta}^\varsigma) =\mu_2 \big(2z_{n,u}-4z_{n,u}^3 \big) - 4 \mu_3 \big(0, (\theta_{n,u}^\varsigma)^2\nonumber\\
&\hspace{1em}- z_{n,u}^2 \big)z_{n,u}- 4 \mu_3 \max \Big(0,1-\sum\nolimits_{n \in \mathcal{N}} z_{n,u}^2 \Big)z_{n,u} \nonumber \\
&+ 4 \mu_4 \bigg[\max \Big(0,\sum\nolimits_{i \in \mathcal {U}} z_{n,i}^2 \mathrm{SE}_{i}^{\varsigma}(\boldsymbol{\theta}^\varsigma) \nonumber \\
&+ \!\!\sum_{m \in \mathcal{M}}\!\!\bar{z}_{n,m}^2\sum_{k \in \mathcal{K}_m} \!\!\mathrm{SE}_{m,k}^{\varsigma}(\boldsymbol{\theta}^\varsigma) - C_{\mathrm{max},n}\Big)\bigg]z_{n,u} \mathrm{SE}_{u}^{\varsigma}(\boldsymbol{\theta}^\varsigma),
\end{align}
while $\frac{\partial}{\partial \bar{\theta}_{n,m}^\varsigma} g(\boldsymbol{\vartheta}^\varsigma)$ and $\frac{\partial}{\partial \bar{z}_{n,m}} g(\boldsymbol{\vartheta}^\varsigma)$ are given by
\begin{align}
 &\frac{\partial}{\partial \bar{\theta}_{n,m}^\varsigma} g(\boldsymbol{\vartheta}^\varsigma)=\nonumber\\
 & -w_2\sum_{i \in \mathcal{M}} \sum_{k \in \mathcal{K}_m} \frac{\partial}{\partial \bar{\theta}_{n,m}^\varsigma} \mathrm{SE}_{i,k}^{\varsigma}(\boldsymbol{\theta}^\varsigma) + X \frac{\partial}{\partial \bar{\theta}_{n,m}^\varsigma} \bar{C}_{m} (\boldsymbol{\vartheta}^\varsigma),
 \end{align}
 \begin{align}
&\frac{\partial}{\partial \bar{z}_{n,m}} g(\boldsymbol{\vartheta}^\varsigma)=\nonumber\\
&-w_2\sum_{i \in \mathcal{M}} \sum_{k \in \mathcal{K}_m} \frac{\partial}{\partial \bar{z}_{n,m}} \mathrm{SE}_{i,k}^{\varsigma} (\boldsymbol{\theta}^\varsigma) +X \frac{\partial}{\partial \bar{z}_{n,m}} \bar{C}_{m} (\boldsymbol{\vartheta}^\varsigma),
\end{align}
respectively, where $\bar{C}_m(\boldsymbol{\vartheta}^\varsigma)= \mu_1 \bar{C}_{1,m}(\boldsymbol{\theta}^\varsigma)+\mu_2 \bar{C}_{2,m}(\mathbf{z})+\mu_3 \bar{C}_{3,m}(\boldsymbol{\theta}^\varsigma, \mathbf{z}) + \mu_4 C_{4}(\boldsymbol{\theta}^\varsigma, \mathbf{z})$. 
On the other hand $\frac{\partial}{\partial \bar{\theta}_{n,m}^\varsigma}\mathrm{SE}_{i,k}^{\varsigma}(\boldsymbol{\theta}^\varsigma)$ is calculated as
\begin{align}\label{eq:gradthanm}
\frac{\partial}{\partial \bar{\theta}_{n,m}^\varsigma} & \mathrm{SE}_{i,k}^{\varsigma}(\boldsymbol{\theta}^\varsigma)=\nonumber\\
&\hspace{-2.5em}\frac{T-\tau}{T\ln 2}\!\bigg[\!\frac{\frac{\partial}{\partial \bar{\theta}_{n,m}^\varsigma}(U_{i,k}^\varsigma(\boldsymbol{\theta}^\varsigma)+V_{i,k}^\varsigma(\boldsymbol{\theta}^\varsigma))}{(U_{i,k}^\varsigma(\boldsymbol{\theta}^\varsigma)+V_{i,k}^\varsigma(\boldsymbol{\theta}^\varsigma))}\!-\!\frac{\frac{\partial}{\partial \bar{\theta}_{n,m}^\varsigma}V_{i,k}^\varsigma(\boldsymbol{\theta}^\varsigma)}{V_{i,k}^\varsigma(\boldsymbol{\theta}^\varsigma)}\bigg],
\end{align}
with
\begin{align}    
&\frac{\partial}{\partial \bar{\theta}_{n,m}^\varsigma} U_{i,k}^\varsigma(\boldsymbol{\theta}^\varsigma)\!=\! 
\begin{cases}
\!2p_{\text{dl}}(\sum\limits_{n \in \mathcal{N}} \bar{\theta}_{n,m}^\varsigma \bar{\Lambda}_{n,m,k}^\varsigma)(\bar{\Lambda}_{n,m,k}^\varsigma), & i\!=\!m, \\
\!0, & i \!\neq \!m, 
\end{cases} \nonumber\\
&\frac{\partial}{\partial \bar{\theta}_{n,m}^\varsigma} V_{i,k}^\varsigma(\boldsymbol{\theta}^\varsigma)= 
\begin{cases}
2~ p_{\text{dl}} \bar{\theta}_{n,m}^\varsigma \bar{\Theta}_{n,m,k}^\varsigma, & i = m , \\
2~ p_{\text{dl}} \bar{\theta}_{n,m}^\varsigma \bar{\Theta}_{n,i,k}^\varsigma, & i \neq m,
\end{cases}  
\end{align}
while
$ -\sum_{i \in \mathcal{M}} \sum_{k \in \mathcal{K}_m} \frac{\partial}{\partial \bar{z}_{n,m}} \mathrm{SE}_{i,k}^\varsigma(\boldsymbol{\theta}^\varsigma) = 0, ~ \forall~ n,i,k$. In addition,
\begin{align}
&\frac{\partial}{\partial \bar{\theta}_{n,m}^\varsigma} \bar{C}_{m} (\boldsymbol{\vartheta}^\varsigma)
\!=\! - \mu_1 \!\!\!\sum_{m \in \mathcal{M}} \sum_{k \in \mathcal{K}_m}\!\! 2  \max \big(0,\mathrm{\bar{SE}}_{QoS} \!-\!\mathrm{SE}_{i}^{\varsigma}(\boldsymbol{\theta}^\varsigma)\big)\nonumber \\
&\times\frac{\partial}{\partial \bar{\theta}_{n,m}^\varsigma} \mathrm{SE}_{i}^{\varsigma}(\boldsymbol{\theta}^\varsigma)+ 4 \mu_3 \max \big(0, (\bar{\theta}_{n,m}^\varsigma)^2 - \bar{z}_{n,m}^2 \big)\bar{\theta}_{n,m}^\varsigma\nonumber \\
&+ 2 \mu_4 \bigg[\max \Big(0,\sum_{u \in \mathcal {U}} z_{n,u}^2 \mathrm{SE}_{u}^{\varsigma}(\boldsymbol{\theta}^\varsigma)\nonumber \\
&+\!\! \sum_{i \in \mathcal{M}}\!\!\bar{z}_{n,i}^2\!\!\sum_{k \in \mathcal{K}_i}\!  \mathrm{SE}_{i,k}^{\varsigma}(\boldsymbol{\theta}^\varsigma) \!- \!C_{\mathrm{max},n}\Big)\bigg] \bar{z}_{n,i}^2 \!\!\sum_{k \in \mathcal{K}_i}\!\! \frac{\partial}{\partial \bar{\theta}_{n,m}^\varsigma} \mathrm{SE}_{i,k}^{\varsigma}(\boldsymbol{\theta}^\varsigma),
\end{align}
and
\begin{align}
\frac{\partial}{\partial \bar{z}_{n,m}} &\bar{C}_{m} (\boldsymbol{\vartheta}^\varsigma) \!\!= \nonumber\\
&\!\! \mu_2 \big(2\bar{z}_{n,m}-4\bar{z}_{n,m}^3\big)\! - \! 4 \mu_3 \big(0, (\bar{\theta}_{n,m}^\varsigma)^2 \!\! - \! \bar{z}_{n,m}^2 \big)\bar{z}_{n,m}
\nonumber\\
&- 4 \mu_3 \max \Big(0,1-\sum_{n \in \mathcal{N}} \bar{z}_{n,m}^2 \Big)\bar{z}_{n,m} \nonumber \\
&\!\!\!\!+  4 \mu_4 \bigg[\max \Big(0,\!\!\sum_{u \in \mathcal {U}}\!\! z_{n,u}^2 \mathrm{SE}_{u}^{\varsigma}(\boldsymbol{\theta}^\varsigma) +\!\!\! \sum_{i \in \mathcal{M}}\!\!\bar{z}_{n,i}^2\!\!\sum_{k \in \mathcal{K}_i} \!\!\mathrm{SE}_{i,k}^{\varsigma}(\boldsymbol{\theta}^\varsigma) \nonumber \\
&- C_{\mathrm{max},n}\Big)\bigg] \bar{z}_{n,m} \sum_{k \in \mathcal{K}_m} \mathrm{SE}_{m,k}^{\varsigma}(\boldsymbol{\theta}^\varsigma).
\end{align}
\begin{algorithm}[!t]
\caption{Solving \eqref{eq:apgopt} Using APG Approach }
\begin{algorithmic}[1]\label{algo}
\STATE \textbf{Initialize}: $X >0$, ${\boldsymbol{\vartheta}^\varsigma}^{(0)}$,  $\alpha_{\overline{\boldsymbol{\vartheta}^\varsigma}}>0, \alpha_{\boldsymbol{\vartheta}^\varsigma}>0$. Set $\tilde{\boldsymbol{\vartheta}^\varsigma}^{(1)}={\boldsymbol{\vartheta}^\varsigma}^{(1)}={\boldsymbol{\vartheta}^\varsigma}^{(0)}, \zeta \in[0,1), b^{(1)}=1, {c^\varsigma}^{(1)} = g({\boldsymbol{\vartheta}^\varsigma}^{(1)})$,   $o=1, q^{(0)}=0, q^{(1)}=1$. Choose $\overline{\boldsymbol{\vartheta}}^{(0)}$ from feasible set $ \widehat{\mathcal{C}}$.
\STATE \textbf{Repeat}  
\STATE \textbf{while}  $\big|\frac{g\left({\boldsymbol{\vartheta}^\varsigma}^{(o)}\right)-g\left({\boldsymbol{\vartheta}^\varsigma}^{(o-10)}\right)}{g\left({\boldsymbol{\vartheta}^\varsigma}^{(o)}\right)}\big|\!\leq  \epsilon$ \text{or}  $\big|\frac{f\left({\boldsymbol{\theta}^\varsigma}^{(o)}\right)-f\left({\boldsymbol{\theta}^\varsigma}^{(o-1)}\right)}{f\left({\boldsymbol{\theta}^\varsigma}^{(o)}\right)}\big|\!\!\leq\!  \epsilon$ \textbf{do}
\STATE \text { update } $\overline{{\boldsymbol{\vartheta}^\varsigma}}^{(o)}$ as~\eqref{eq:Vo}
\STATE Set $\tilde{\boldsymbol{\vartheta}^\varsigma}^{(o+1)}=\mathcal{P}_{\hat{\mathcal{C}}} (\overline{\boldsymbol{\vartheta}^\varsigma}^{(o)}-\alpha_{\boldsymbol{\vartheta}^\varsigma} \nabla g (\overline{\boldsymbol{\vartheta}^\varsigma}^{(o)}))$,
\STATE \textbf{ if } $g(\tilde{\boldsymbol{\vartheta}^\varsigma}^{(o+1)}) \leq {c^\varsigma}^{(o)}-\zeta \|\tilde{\boldsymbol{\vartheta}^\varsigma}^{(o+1)}-\overline{\boldsymbol{\vartheta}^\varsigma}^{(o)} \|^2$ \textbf{then}
\STATE ${\boldsymbol{\vartheta}^\varsigma}^{(o+1)}=\tilde{\boldsymbol{\vartheta}^\varsigma}^{(o+1)}$
\STATE \textbf{else}
\STATE update $\hat{\boldsymbol{\vartheta}^\varsigma}^{(o+1)}$ \text { using \eqref{eq:vnhat+1} } and then ${\boldsymbol{\vartheta}^\varsigma}^{(o+1)}$ \text { using \eqref{eq:vn+1} }
\STATE \textbf{end if}
\STATE update $q^{(o+1)}$ using~\eqref{eq:q_n}.
\STATE update $b^{(o+1)}$ \text { using \eqref{eq:bn+1}  and } ${c^\varsigma}^{(o+1)}$ \text { using \eqref{eq:cn+1} }
\STATE  update $o=o+1$
\STATE \textbf{end while} 
\STATE \textbf{until} {convergence}.
\end{algorithmic}
\end{algorithm}
\subsubsection{Projection onto $\widehat{\mathcal{C}}$}
The projection of the   given  $\mathbf{r}  \in \mathbb{R}^{2N(U+M) \times 1}$  onto the feasible set $\widehat{\mathcal{C}}$ in  \textbf{Step 5} of   \textbf{Algorithm~\ref{algo}}
 can be done by solving the  problem
\begin{align} \label{eq:proj}
\mathcal{P}_{\widehat{\mathcal{C}}}(\mathbf{r}):&\min _{\boldsymbol{\vartheta}^\varsigma \in \mathbb{R}^{2N(U+M) \times 1}}\|\boldsymbol{\vartheta}^\varsigma-\mathbf{r}\|^2 \\
&\text { s.t. } ~~ \eqref{eq:cons2g},\eqref{eq:cons3g},\eqref{eq:znuznm}, \eqref{eq:znuznmkh},
\end{align}
with $\mathbf{r}=\left[\mathbf{r}_1^T, \mathbf{r}_2^T\right]^T$, $\mathbf{r}_1 \triangleq\left[\mathbf{r}_{1,1}^T, \ldots, \mathbf{r}_{1, N}^T\right]^T$ and $\mathbf{r}_{1, n} \triangleq\left[r_{1, n1}, \ldots, r_{1, nU} ,\bar{r}_{1, n1},\ldots,\bar{r}_{1, nM} \right]^T$, while $\mathbf{r}_{2, n} \triangleq\left[r_{2, n1} , \ldots, r_{2, nU} ,\bar{r}_{2, n1},\ldots,\bar{r}_{2, nM} \right]^T$. Problem \eqref{eq:proj} can be split into two separate subproblems for calculating $\boldsymbol{\theta}_n^\varsigma $ and $\mathbf{z}_n$. Following a similar approach as in \cite{Farooq:TCOM:2021}  we can find the following closed-form expression 
\begin{align} \label{eq:prothn}
\boldsymbol{\theta}_n^\varsigma=\frac{\sqrt{\rho}}{\max \left(\sqrt{\rho},\left\|\left[\mathbf{r}_{1, n}\right]_{0}^+\right\|\right)}\left[\mathbf{r}_{1, n}\right]_{0}^+,
\end{align}
where $ \big[\Pi\big]_{0}^+\!\! \triangleq\!\!\big[\max \left(0, \pi_1 \right)\!,\! \ldots\!,\! \max \left(0, \pi_{U} \right)\!,\max \!\left(0, \bar{\pi}_1\right)\!,\dots,\\\max \left(0, \bar{\pi}_{M}\right) \big]^T, \forall~ \Pi \in \mathbb{R}^{(U+M) \times 1}$
and
 \begin{align}
 \label{eq:prozn}
\mathbf{z}_n=\bigg[\frac{\sqrt{K_{\mathrm{max}}}}{\max \left(\sqrt{K_{\mathrm{max}}},\left\|\left[\mathbf{r}_{2, n}\right]_{0}^+\right\|\right)}\left[\mathbf{r}_{2, n}\right]_{0}^+\bigg]_{1-},
\end{align}
where 
$ \big[\Pi\big]_{1-}\! \triangleq\! \big[\min\! \left(1, \pi_1 \right)\!, \ldots, \min \!\left(1, \pi_{U} \right)\!,
\min\! \left(1, \bar{\pi}_1\right)\!,\dots,\\\min \left(1, \bar{\pi}_{M}\right) \big]^T, \forall~ \Pi \in \mathbb{R}^{(U+M) \times 1}.$
Given that the feasible set $\widehat{\mathcal{C}}$ is bounded, it follows that $\nabla g(\boldsymbol{\vartheta}^\varsigma)$ is Lipschitz continuous with a known constant $J$. This implies that for all $\mathbf{v}, \mathbf{w} \in \widehat{\mathcal{C}}$, the gradient satisfies $\|\nabla g(\mathbf{v})-\nabla g(\mathbf{w})\| \leq J\|\mathbf{v}-\mathbf{w}\|$.

In \textbf{Algorithm~\ref{algo}}, beginning with a random initial point $\overline{\boldsymbol{\vartheta}^\varsigma}^{(0 )}$, we update $\overline{\boldsymbol{\vartheta}^\varsigma}^{(o )}$ at each iteration as follows:
\begin{align} \label{eq:Vo}
&\overline{\boldsymbol{\vartheta}^\varsigma}^{(o)} =\nonumber\\
& {\boldsymbol{\vartheta}^\varsigma}^{(o)}\! +\! \frac{q^{(o - 1)}}{q^{(o)}}\left(\!\tilde{\boldsymbol{\vartheta}^\varsigma}^{(o)}\! -\! {\boldsymbol{\vartheta}^\varsigma}^{(o)}\!\right)\! + \! \frac{q^{(o - 1)}\! -\! 1}{q^{(o)}}\left(\!{\boldsymbol{\vartheta}^\varsigma}^{(o)}\! -\! {\boldsymbol{\vartheta}^\varsigma}^{(o - 1)}\!\right),
\end{align}
where 
\begin{align} \label{eq:q_n}
 q^{(o+1)}=\frac{1+\sqrt{4\left(q^{(o)}\right)^2+1}}{2}.
 \end{align}
We then proceed along the gradient of the objective function with a specified step size $\alpha_{\overline{\boldsymbol{\vartheta}^\varsigma}}$. The resulting point $\left(\overline{\boldsymbol{\vartheta}^\varsigma}-\alpha_{\overline{\boldsymbol{\vartheta}^\varsigma}} \nabla g(\overline{\boldsymbol{\vartheta}^\varsigma})\right)$ is subsequently projected onto the feasible set $\widehat{\mathcal{C}}$, yielding
\begin{align} 
\tilde{\boldsymbol{\vartheta}^\varsigma}^{(o+1)}=\mathcal{P}_{\hat{\mathcal{C}}}\left(\overline{\boldsymbol{\vartheta}^\varsigma}^{(o)}-\alpha_{\overline{\boldsymbol{\vartheta}^\varsigma}} \nabla g\left(\overline{\boldsymbol{\vartheta}^\varsigma}^{(o)}\right)\right).
\end{align}
It is important to note that $g(\boldsymbol{\vartheta}^\varsigma)$ is not convex, so $g\big(\tilde{\boldsymbol{\vartheta}^\varsigma}^{(o+1)}\big)$ may not necessarily improve the objective sequence. Consequently, ${\boldsymbol{\vartheta}^\varsigma}^{(o+1)} = \tilde{\boldsymbol{\vartheta}^\varsigma}^{(o+1)}$ is accepted only if the objective value $g\big(\tilde{\boldsymbol{\vartheta}^\varsigma}^{(o+1)}\big)$ is below $c^{(o)}$, which acts as a relaxation of $g\big({\boldsymbol{\vartheta}^\varsigma}^{(o)}\big)$ while remaining relatively close to it. Moreover,
$c^{(o)}$ can be computed as 
\begin{align} 
& {c^\varsigma}^{(o)}= \frac{\sum\nolimits_{o=1}^\kappa \zeta^{(\kappa-o)} g\big({\boldsymbol{\vartheta}^\varsigma}^{(o)}\big)}{\sum\nolimits_{o=1}^\kappa \zeta^{(\kappa-o)}},
\end{align}
where $\zeta \in[0,1)$ is used to control the non-monotonicity degree. 
After each iteration, ${c^\varsigma}^{(o)}$ can be updated iteratively as follows
\begin{align}\label{eq:cn+1}
&{c^\varsigma}^{(o+1)}=\frac{\zeta b^{(o)} {c^\varsigma}^{(o)}+g\big({\boldsymbol{\vartheta}^\varsigma}^{(o)}\big)}{b^{(o+1)}},
\end{align}
where ${c^\varsigma}^{(1)}\!\!=\!\!g\big({\boldsymbol{\vartheta}^\varsigma}^{(1)}\big)$ and $b^{(1)}\!\!=\!\!1$, and ${b}^{(o+1)}$ can be  obtained as
\begin{align}\label{eq:bn+1}
&{b}^{(o+1)}=\zeta b^{(o)}+1. 
\end{align}
When the condition $g \big(\tilde{\boldsymbol{\vartheta}^\varsigma}^{(o+1)}\big) \leq {c^\varsigma}^{(o)}-\zeta\Vert\tilde{\boldsymbol{\vartheta}^\varsigma}^{(o+1)}-\overline{\boldsymbol{\vartheta}^\varsigma}^{(o)}\Vert^2$ is not satisfied, extra correction steps are employed to avoid this situation. Specifically, another point is calculated with a dedicated step size $\alpha_{\boldsymbol{\vartheta}}$ as 
\begin{align} \label{eq:vnhat+1}
\hat{\boldsymbol{\vartheta}^\varsigma}^{(o+1)}=\mathcal{P}_{\hat{\mathcal{C}}}\left({\boldsymbol{\vartheta}^\varsigma}^{(o)}-\alpha_{\boldsymbol{\vartheta}^\varsigma} \nabla g\big({\boldsymbol{\vartheta}^\varsigma}^{(o)}\big)\right).
\end{align}
Then, ${\boldsymbol{\vartheta}^\varsigma}^{(o+1)}$ is updated by comparing the objective values at $\tilde{\boldsymbol{\vartheta}^\varsigma}^{(o+1)}$ and $\hat{\boldsymbol{\vartheta}^\varsigma}^{(o+1)}$ as
\begin{align}\label{eq:vn+1}
\!{\boldsymbol{\vartheta}^\varsigma}^{(o+1)} \triangleq\!\left\{\begin{array}{ll}
\tilde{\boldsymbol{\vartheta}^\varsigma}^{(o+1)}, &\!\! \text { if } g\big(\tilde{\boldsymbol{\vartheta}^\varsigma}^{(o+1)}\big) \leq g\big(\hat{\boldsymbol{\vartheta}^\varsigma}^{(o+1)}\big),  \\
\hat{\boldsymbol{\vartheta}^\varsigma}^{(o+1)}, & \text {otherwise}.
\end{array}\right.
\end{align}
Finally, we emphasize that our proposed APG-based optimization approach operates on the large-scale fading timescale, which varies slowly over time~\cite{zahra:MSP:2025}.
\begin{algorithm}[!t]
\caption{Solving \eqref{eq:scaoptf} Using SCA Approach}
\begin{algorithmic}[1]\label{algo2}
\STATE \textbf{Initialize}: iteration index $i=0$,  $\lambda > 1$, $\tilde{\mathbf{v}}^{\varsigma^{(0)}} \in \tilde{\mathcal{V}}$.
\STATE \textbf{repeat}
\STATE \text{ update } $i=i+1$
\STATE \text{ solve} \eqref{eq:scaoptf} \text {to obtain the optimal} $\tilde{\mathbf{v}}^{\varsigma *}$
\STATE \text{ update } $\tilde{\mathbf{v}}^{\varsigma^{(i)}}=\tilde{\mathbf{v}}^{\varsigma *}$ 
\STATE \textbf{until} {convergence}.
\end{algorithmic}
\end{algorithm}
\subsection{SCA-Based Optimization Approach}
The APG method is characterized by low computational complexity. However, its effectiveness depends on the careful selection of the step sizes, $\alpha_{\overline{\boldsymbol{\vartheta}^\varsigma}}$ and $\alpha_{\boldsymbol{\vartheta}}$, to ensure convergence. Moreover, it does not guarantee a globally optimal solution for non-convex optimization problems and requires additional algorithmic development to improve its performance in such scenarios \cite{hao2024joint,Li:2015:NIPS}.
Here, for  completeness. we provide the SCA-based method to solve the optimization problem in \eqref{eq:optpg}. This approach can offer better performance due to its ability to find a globally stationary solution point for non-convex problems, albeit at the cost of increased computational complexity. The joint optimization problem can be reformulated as
\begin{subequations}\label{eq:scaopt}
\begin{align}
&\min\limits_{\tilde{\mathbf{v}}^\varsigma} \tilde{\mathbf{g}}(\tilde{\mathbf{v}}^\varsigma)  ,\\
 &\text {s.t.} ~  \eqref{eq:cons2g},\eqref{eq:cons3g},\eqref{eq:cons5g},\eqref{eq:cons6g},\eqref{eq:tha},\eqref{eq:vw}, \nonumber\\
 &0 \leq a_{n,u} \leq 1 , 0 \leq \bar{a}_{n,m} \leq 1,~ \forall n,u,m, \label{eq:sconstan}
 \\
&t_u^\varsigma \leq  \mathrm{SE}_{u}^{\varsigma} (\boldsymbol{\theta}^\varsigma) ~, ~ \bar{t}_{m,k}^\varsigma \leq  \mathrm{SE}_{m,k}^{\varsigma} (\boldsymbol{\theta}^\varsigma), ~~\forall ~ u,m,k, \label{eq:sconstse} \\
& t_u^\varsigma \ge  \mathrm{SE}_{QoS} ~, ~ \bar{t}_{m,k}^\varsigma \ge  \bar{\mathrm{SE}}_{QoS}, ~~\forall ~ u,m,k, \label{eq:sconstqs} \\
& \sum_{u \in \mathcal {U}} a_{n,u} \hat{t}_u^\varsigma + \sum_{m \in \mathcal{M}}\bar{a}_{n,m}\sum_{k \in \mathcal{K}_m}  \hat{\bar{t}}_{m,k}^\varsigma \leq C_{\mathrm{max},n},~~ \forall ~ n, \label{eq:sconstcm}\\
& \widetilde{\mathrm{SE}}_{u}^{\varsigma} (\boldsymbol{\theta}^\varsigma,\mathbf{w}^\varsigma) \leq \hat{t}_u^\varsigma ~,~ \widetilde{\mathrm{SE}}_{m,k}^{\varsigma} (\boldsymbol{\theta}^\varsigma,\mathbf{w}^\varsigma)\! \leq \!\hat{\bar{t}}_{m,k}^\varsigma, ~\forall  u,m,k, \label{eq:sconsset}
\end{align}
\end{subequations}
where $\tilde{\mathbf{v}}^\varsigma \triangleq \{\boldsymbol{\theta}^\varsigma, \mathbf{a},\mathbf{w}^\varsigma,\hat{\mathbf{t}}^\varsigma,\mathbf{t}^\varsigma\}$, $\tilde{\mathbf{g}}(\tilde{\mathbf{v}}^\varsigma)\triangleq -\Big(w_1 \sum_{u \in \mathcal {U}} t_u^\varsigma  + w_2 \sum_{m \in \mathcal{M}} \sum_{k \in \mathcal{K}_m} \bar{t}_{m,k}^\varsigma \Big) + \lambda \Big(S_u(\mathbf{a})+\bar{S}_{m}(\mathbf{a})\Big) $, and $\lambda$ is the Lagrangian multiplier. Further,  $\mathbf{w}^\varsigma \triangleq \{\mathrm{w}_u^\varsigma, \mathrm{w}_{m,k}^\varsigma\}$ is the new variable given as
\begin{align}\label{eq:vw}
&V_u^\varsigma(\boldsymbol{\theta}^\varsigma) \ge \mathrm{w}_u^\varsigma, \nonumber\\
& V_{m,k}^\varsigma(\boldsymbol{\theta}^\varsigma) \ge \mathrm{w}_{m,k}^\varsigma.
\end{align}
$\hat{\mathbf{t}}^\varsigma \triangleq\{\hat{t}_u^\varsigma, \hat{\bar{t}}_{m,k}^\varsigma\}$ and $\mathbf{t}^\varsigma\triangleq\{t_u^\varsigma, \bar{t}_{m,k}^\varsigma\}$ are the auxiliary variables defined in \eqref{eq:sconstse} and \eqref{eq:sconsset}, respectively.
Moreover, we have 
\begin{align}\label{eq:seuvt}
\mathrm{SE}_{u}^{\varsigma} (\boldsymbol{\theta}^\varsigma) \leq \widetilde{\mathrm{SE}}_{u}^{\varsigma} (\boldsymbol{\theta}^\varsigma,\mathbf{w}^\varsigma) \triangleq \frac{T-\tau}{T} \log_{2}\left({1+\frac{ (U_u^\varsigma(\boldsymbol{\theta}^\varsigma))^2}{\mathrm{w}_u^\varsigma}} \right), 
\end{align}
and
\begin{align}\label{eq:semkvt}
    \mathrm{SE}_{m,k}^{\varsigma} (\boldsymbol{\theta}^\varsigma)\! \leq\! \widetilde{\mathrm{SE}}_{m,k}^{\varsigma} (\boldsymbol{\theta}^\varsigma\!,\mathbf{w}^\varsigma) \!\triangleq\!\! \frac{T-\tau}{T} \log_{2}\!\left(\!\!{1\!+\!\frac{(U_{m,k}^\varsigma(\boldsymbol{\theta}^\varsigma))^2}{\mathrm{w}_{m,k}^\varsigma}} \!\right).
\end{align}
Following the same approach as in \cite{hao2024joint} and \cite{Mohammadi:JSAC:2024}, the convex lower bound of $\mathrm{SE}_{u}^{\varsigma}(\boldsymbol{\theta}^\varsigma)$ for unicast users in \eqref{eq:seuvt} is given by:
\begin{align}
\hat{\mathrm{SE}}_{u}^{\varsigma} (\boldsymbol{\theta}^\varsigma)\ &\triangleq  \frac{T-\tau}{T \log 2} \Bigg[ \log \bigg(1+\frac{\big((U_u^\varsigma)^{(i)}\big)^2}{(\mathrm{w}_u^\varsigma)^{(i)}}\bigg)  -\frac{\big((U_u^\varsigma)^{(i)}\big)^2}{(\mathrm{w}_u^\varsigma)^{(i)}} \nonumber\\ 
& + 2 \frac{(U_u^\varsigma)^{(i)} U_u^\varsigma}{(\mathrm{w}_u^\varsigma)^{(i)}} -\frac{\big((U_u^\varsigma)^{(i)}\big)^2\big((U_u^\varsigma)^2+\mathrm{w}_u^\varsigma\big)}{(\mathrm{w}_u^\varsigma)^{(i)}\Big(\big((U_u^\varsigma)^{(i)}\big)^2+(\mathrm{w}_u^\varsigma)^{(i)}\Big)} \Bigg].
\end{align}
Similarly, for multicast uses, the convex lower bound of $\mathrm{SE}_{m,k}^{\varsigma} (\boldsymbol{\theta}^\varsigma)$ in \eqref{eq:semkvt} is
\vspace{-0.5em}
\begin{align}
\hat{\mathrm{SE}}_{m,k}^{\varsigma}  &\triangleq  \frac{T-\tau}{T \log 2} \Bigg[ \log \bigg(1+\frac{\big((U_{m,k}^\varsigma)^{(i)}\big)^2}{(\mathrm{w}_{m,k}^\varsigma)^{(i)}}\bigg)-\frac{\big((U_{m,k}^\varsigma)^{(i)}\big)^2}{(\mathrm{w}_{m,k}^\varsigma)^{(i)}} \nonumber \\ 
& \hspace{-1cm}+\! 2 \frac{(U_{m,k}^\varsigma)^{(i)} U_{m,k}^\varsigma}{(\mathrm{w}_{m,k}^\varsigma)^{(i)}} \!-\!\frac{\big((U_{m,k}^\varsigma)^{(i)}\big)^2\big((U_{m,k}^\varsigma)^2+\mathrm{w}_{m,k}^\varsigma\big)}{(\mathrm{w}_{m,k}^\varsigma)^{(i)}\Big(\big(U_{m,k}^{(i)}\big)^2\!+\!(\mathrm{w}_{m,k}^\varsigma)^{(i)}\Big)} \Bigg].
\end{align}
We point out that $\widetilde{\mathrm{SE}}_{u}^{\varsigma} (\boldsymbol{\theta}^\varsigma,\mathbf{w}^\varsigma)$ in \eqref{eq:seuvt} and $\widetilde{\mathrm{SE}}_{m,k}^{\varsigma} (\boldsymbol{\theta}^\varsigma,\mathbf{w}^\varsigma)$ in \eqref{eq:semkvt} have convex upper bounds, given respectively by

\begin{align}
\bar{\mathrm{SE}}_{u}^{\varsigma} (\boldsymbol{\theta}^\varsigma,\mathbf{w}^\varsigma) &\triangleq  \frac{T-\tau}{T \log 2} \bigg[ \log \Big(\big((U_u^\varsigma)^{(i)}\big)^2 \nonumber \\
&\hspace{-3.3em}+\! (\mathrm{w}_u^\varsigma)^{(i)}\!\Big)\!\! +\! \frac{ (U_u^\varsigma)^2 + \mathrm{w}_u^\varsigma }{\big((U_u^\varsigma)^{(i)}\big)^2+(\mathrm{w}_u^\varsigma)^{(i)}}\! -\! 1\! -\! \log(\mathrm{w}_u^\varsigma) \bigg], 
\end{align}
and
\begin{align}
\bar{\mathrm{SE}}_{m,k}^{\varsigma} (\boldsymbol{\theta}^\varsigma,\mathbf{w}^\varsigma) & \triangleq  \frac{T-\tau}{T \log 2} \bigg[ \log \Big(\big((U_{m,k}^\varsigma)^{(i)}\big)^2 + (\mathrm{w}_{m,k}^\varsigma)^{(i)}\Big) \nonumber \\
&\hspace{-4em}+ \frac{ (U_{m,k}^\varsigma)^2 + \mathrm{w}_{m,k}^\varsigma }{\big((U_{m,k}^\varsigma)^{(i)}\big)^2+(\mathrm{w}_{m,k}^\varsigma)^{(i)}} - 1 - \log(\mathrm{w}_{m,k}^\varsigma) \bigg].
\end{align}
Using the first-order Taylor series expansion, we can obtain a convex upper bound for $S_u(\mathbf{a})$ and $\bar{S}_{m}(\mathbf{a})$, respectively, as
\begin{align}
\widehat{S}_u(\mathbf{a}) 
\triangleq
 \sum_{u \in \mathcal {U}} \sum_{n \in \mathcal{N}} \big(a_{n,u}- 2 a_{n,u} a_{n,u}^{(i)} +(a_{n,u}^{(i)})^2\big),
\end{align}
\vspace{-0.3em}
and
\begin{align}
 \widehat{\bar{S}}_{m}(\mathbf{a}) 
\triangleq  \sum_{m \in \mathcal{M}} \sum_{n \in \mathcal{N}} \big(\bar{a}_{n,m} - 2 \bar{a}_{n,m} \bar{a}_{n,m}^{(i)} + (\bar{a}_{n,m}^{(i)})^2\big).
\end{align}
Similarly, we can get a convex upper bound for \eqref{eq:sconstcm} using the first-order Taylor series expansion, as

\begin{align}\label{eq:scacmax}
    & 0.25\! \Big[\!\! \sum_{u \in \mathcal {U}}\!\! \big(\!a_{n,u}\! +\! \hat{t}_u^\varsigma\big)^2\!\!\!\! - \!\!2 \big(   a_{n,u}^{(i)} \!\!\!-\! (\hat{t}_u^\varsigma)^{(i)}\big) \!\big(a_{n,u}\!\!\! - \hat{t}_u^\varsigma\big) \!\!+\!\! \big(   a_{n,u}^{(i)}\!\!\! -\! (\hat{t}_u^\varsigma)^{(i)}\!\big)^2   \nonumber\\
  & +  \sum_{m \in \mathcal{M}} \big(   \bar{a}_{n,m} + \sum_{k \in \mathcal{K}_m} \hat{\bar{t}}_{m,k}^\varsigma\big)^2 - 2 \big(   \bar{a}_{n,m}^{(i)} - \sum_{k \in \mathcal{K}_m} (\hat{\bar{t}}_{m,k}^\varsigma)^{(i)}\big) \nonumber\\
  &\big(   \bar{a}_{n,m} - \sum_{k \in \mathcal{K}_m} \hat{\bar{t}}_{m,k}^\varsigma\big) + \big(   \bar{a}_{n,m}^{(i)} - \sum_{k \in \mathcal{K}_m} (\hat{\bar{t}}_{m,k}^\varsigma)^{(i)}\big)^2  \Big] \leq C_{\mathrm{max},n}.
\end{align}
Finally, problem \eqref{eq:scaopt} is approximated by the following convex problem
\vspace{- 0.5 em}
\begin{align}\label{eq:scaoptf}
\min\limits_{\tilde{\mathbf{v}}^\varsigma \in \tilde{\mathcal{V}}} \tilde{\mathcal{G}}(\tilde{\mathbf{v}}^\varsigma),
\end{align}
where $\tilde{\mathcal{V}} \triangleq\{\eqref{eq:cons2g},\eqref{eq:cons3g},\eqref{eq:cons5g},\eqref{eq:cons6g},\eqref{eq:tha},\eqref{eq:sconstan},\eqref{eq:sconstse},\eqref{eq:sconstqs},\\ \eqref{eq:sconsset},\eqref{eq:vw},\eqref{eq:scacmax}\}$ and $\tilde{\mathcal{G}}(\tilde{\mathbf{v}^\varsigma}) \triangleq -\big(w_1 \sum_{u \in \mathcal {U}} t_u^\varsigma  + w_2 \sum_{m \in \mathcal{M}} \sum_{k \in \mathcal{K}_m} \bar{t}_{m,k}^\varsigma \big) + \lambda \big(\widehat{S}_u(\mathbf{a})+\widehat{\bar{S}}_{m}(\mathbf{a})\big)$.  \textbf{Algorithm~\ref{algo2}} is introduced to solve problem \eqref{eq:scaoptf} and obtain the optimized solution $\tilde{\mathbf{v}}^{\varsigma *}$.
 \subsection{Computational Complexity}\label{complex}
The main components of the APG-based method in \textbf{Algorithm~\ref{algo}} are the gradient of the objective function and the projection function, as defined in \eqref{eq:gradthanu}, \eqref{eq:gradthanm}, and \eqref{eq:prothn}, \eqref{eq:prozn}, respectively. The computational complexity of the gradient function is $\mathcal{O}(N(U+K_M)^2)$, while the projection function has a complexity of $\mathcal{O}(N(U+K_M))$, where $K_M$ denotes the total number of multicast users. Therefore, the upper bound on the per-iteration complexity is $\mathcal{O}(N(U+K_M)^2)$.

In contrast, the complexity of the SCA-based method in \textbf{Algorithm~\ref{algo2}}, is $\mathcal{O}(\sqrt{A_l + A_q}(A_v + A_l + A_q)A_v^2)$, where $A_v \triangleq 2N(U+K_M) + (U + K_M)$, $A_l \triangleq 3N(U+K_M) + N + 4U + 4K_M$, and $A_q \triangleq 2N + N(U+K_M) + U + K_M$ \cite{Mohammadi:JSAC:2024}. Consequently, the SCA-based approach exhibits greater computational complexity than the APG-based method, and requiring more time to solve the optimization problem.
Table \ref{table2} reports the average running time for the optimization problem \eqref{eq:optpg}, using both the APG-based and SCA-based methods implemented on the same computing platform. For MR  with $N = 50$, $U = 7$, and $K_M = 48$, the APG-based method is approximately 10 times faster than the SCA-based method. Moreover, under a larger system configuration with $N = 200$, $U = 7$, and $K_M = 48$, the APG-based method is around 61 times faster. Additionally, the APG-based method achieves a 42-fold speedup when $N = 200$, $U = 10$, and $K_M = 60$. These results demonstrate the APG approach significantly outperforms that the SCA approach in terms of computational efficiency, especially in large-scale systems.
\begin{table}[t]
\centering
\caption{Comparison of average running time (seconds)}
\begin{tabular}{|c|c|c|}
\hline
\textbf{System setup} & \textbf{APG time}  & \textbf{SCA time} \\ 
\hline 
$\MR$, $N\!=\!50$, $U\!=\!7$, $K_M\!=\!48$ & $11.08 ~\mathrm{s}$ & $113.21 ~ \mathrm{s}$ \\
\hline
$\MR$, $N\!=\!200$, $U\!\!=\!7$, $K_M\!=\!48$ & $17.72 ~ \mathrm{s}$ & $1077.3 ~ \mathrm{s}$ \\
\hline
$\MR$, \!$N\!=\!200$, \!$U\!\!=\!10$, $K_M\!=\!60$ & $46.6 ~ \mathrm{s}$ & $1924 ~ \mathrm{s}$ \\
\hline
\end{tabular}
\label{table2}
\end{table}

\textit{\textbf{Remark 2:} Note that in this work, we assume imperfect CSI, where precoding vectors are computed based on instantaneous CSI including estimation errors, while the resource allocation optimization approaches rely solely on statistical CSI. The achievable SE expressions are derived using the “use-and-then-forget” bounding technique, which inherently captures the impact of channel estimation errors. This approach ensures that performance degradation due to CSI imperfections is already reflected in the analysis. Furthermore, in practical scenarios with more severe estimation inaccuracy or latency, robustness can be enhanced through techniques, such as advanced pilot design, adaptive CSI update strategies, or robust optimization formulations \cite{Wang:TSP:2025}.}

\section{Numerical Results}
\subsection{Parameters and Network Setup} \label{result1}
We assume that there are $N$ APs, each equipped with  $L=12$ antennas, to simultaneously  serve $U$ unicast users and $M = 4$ multicast groups, each consisting of $K_m$ users, while all the users and APs are  randomly distributed within an area of size $1 \times 1$ km$^2$.    The pilot length is $\tau=U+M$, while the bandwidth is set to $B=20$ MHz. The large-scale fading coefficients $\beta_{n,u}$ and $\bar{\beta}_{n,m,k}$ are modeled as  \cite{Björnson:TWC:2020}
$\beta_{n,u} = 10^{\frac{\text{PL}_{n,u}^d}{10}}10^{\frac{F_{n,u}}{10}}$ and
$\bar{\beta}_{n,m,k} = 10^{\frac{\text{PL}_{n,k_m}^d}{10}}10^{\frac{F_{n,k_m}}{10}}$, respectively, where $10^{\frac{\text{PL}_{n,u}^d}{10}}$ and $10^{\frac{\text{PL}_{n,k_m}^d}{10}}$ are the path loss, $10^{\frac{F_{n,u}}{10}}$ and $10^{\frac{F_{n,k_m}}{10}}$ denote the shadowing effect with $F_{n,u}\in\mathcal{N}(0,4^2)$ and $F_{n,k_m}\in\mathcal{N}(0,4^2)$ (in dB) for unicast and multicast users, respectively. Also, $\text{PL}_{n,u}^d$ and $\text{PL}_{n,k_m}^d$ are in dB and can be calculated as
$\text{PL}_{n,u}^d = -30.5-36.7\log_{10}\Big(\frac{d_{n,u}}{1\,\text{m}}\Big)$
and
$\text{PL}_{n,k_m}^d = -30.5-36.7\log_{10}\Big(\frac{d_{n,k_m}}{1\,\text{m}}\Big)$ \cite{Björnson:TWC:2020}.
The correlation among the shadowing terms from the $n$-th AP to different $g \in \mathcal{U} \cup \mathcal{M}$  unicast and multicast users can be given by  
\begin{align}
	 \mathbb{E}\{F_{n,g}F_{j,g'}\} \triangleq \begin{cases} 4^22^{-\upsilon_{g,g'}/9\,\text{m}}, & j = n, \\ 0, & \text{otherwise}, \end{cases}
\end{align}
where $\upsilon_{g,g'}$ is the physical distance between users $g$ and $g'$. 
The maximum transmission power for each AP is $p_{\text{dl}} = 1$ W, and for each user  is $p_{\text{ul}} = 0.1$ W, while the noise power is $-92$ dBm.
The minimum SE requirements for the  unicast users and multicast users are chosen to be $\bar{SE}_{QoS}=SE_{QoS}=0.2$ bit/s/Hz. Unless otherwise stated, we set $w_1=w_2=0.5$. 
\vspace{-0.5em}
\subsection{Results and Discussions}\label{result2}

\begin{figure}[t]
\centering 
\begin{subfigure}[b]{0.5\textwidth}
\centering 
\includegraphics[width=1\textwidth]{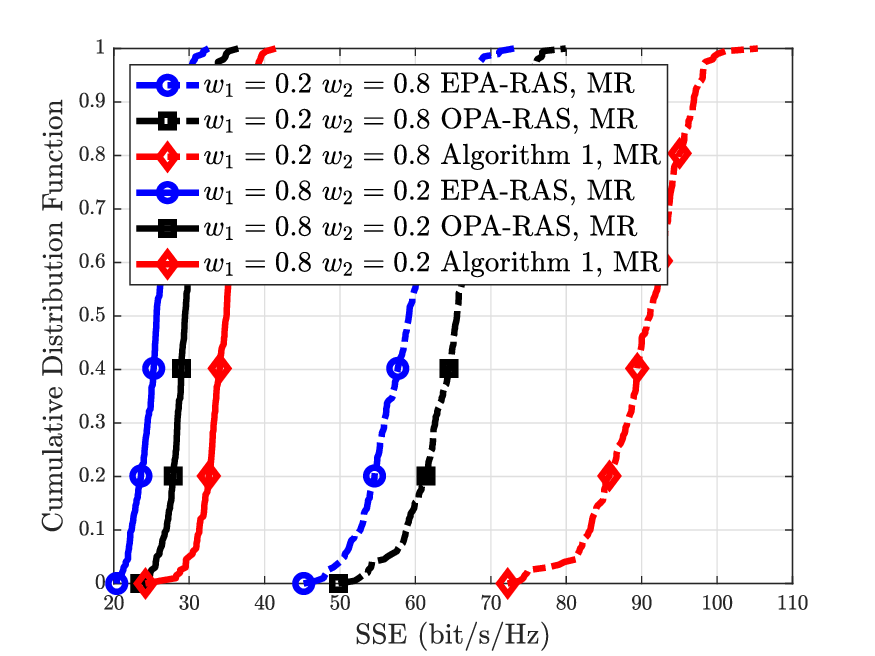}
\caption{MR}
\label{fig2a}
\end{subfigure}
 
\begin{subfigure}[b]{0.5\textwidth}
\centering 
\includegraphics[width=1\textwidth]{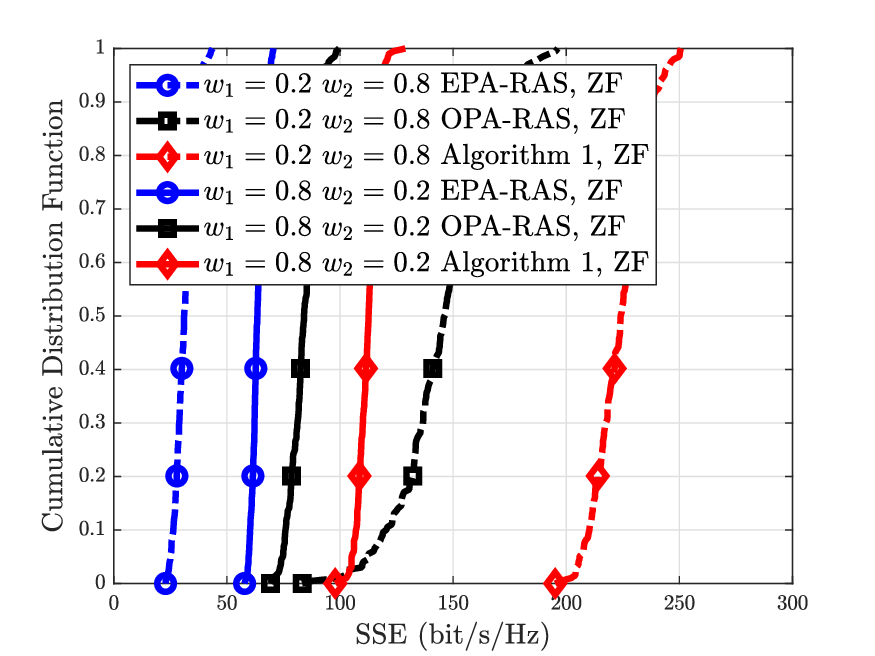}
\caption{ZF}
\label{fig2b}
\end{subfigure}
\centering 
 \caption{CDF of the SSE, where $U=7$,  $K_m=12$, $N=60$.}
 \label{fig2}
\end{figure}
\begin{figure}[t] 
\centering
\includegraphics[width=0.5\textwidth]{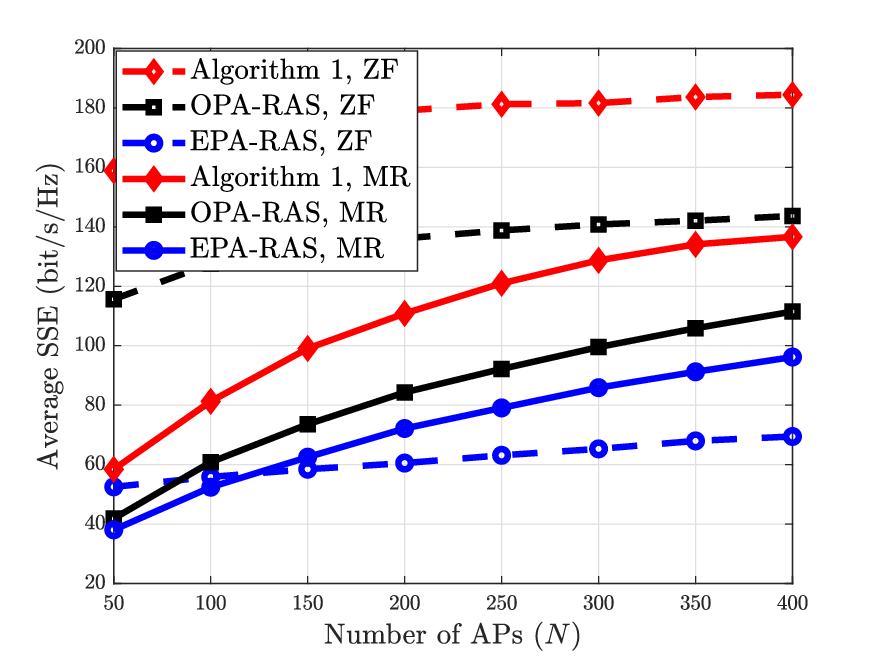}
\centering
\caption{Average SSE against the number of APs $N$, where $U=7$,  $K_m=12$.}
\label{fig3}
\end{figure}

\begin{figure}[t] 
\centering
\includegraphics[width=0.5\textwidth]{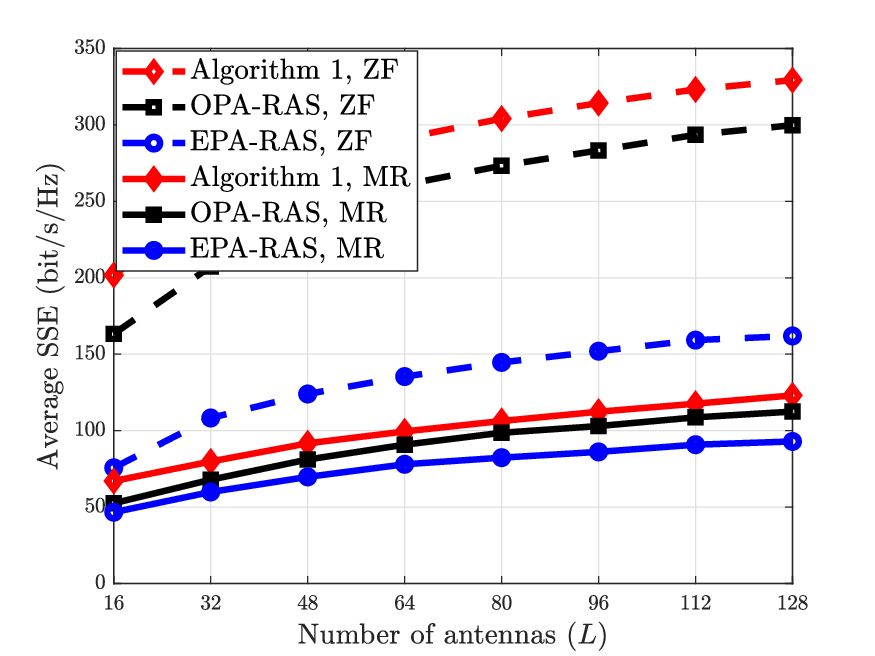}
\centering
\caption{Average SSE against the number of antennas $L$, where, $N=60$, $U=7$,  $K_m=12$.}
\label{fig4}
\end{figure}

\begin{figure}[t] 
\centering
\includegraphics[width=0.5\textwidth]{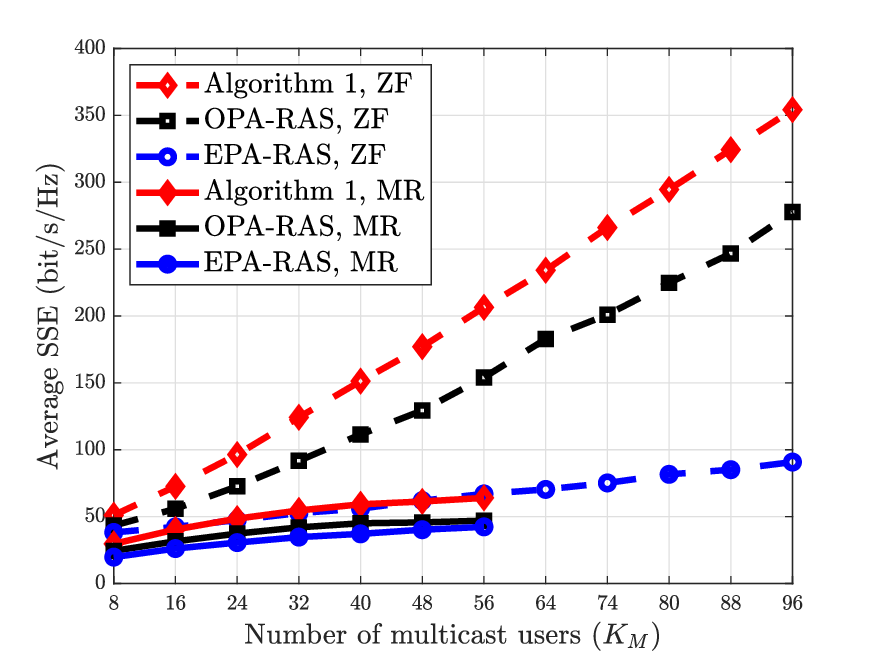}
\centering
\caption{Average  SSE against the number of multicast users,  where  $U=5$, $N=60$.}\label{fig5}
\end{figure}

\begin{figure}[t] 
\centering
\includegraphics[width=0.5\textwidth]{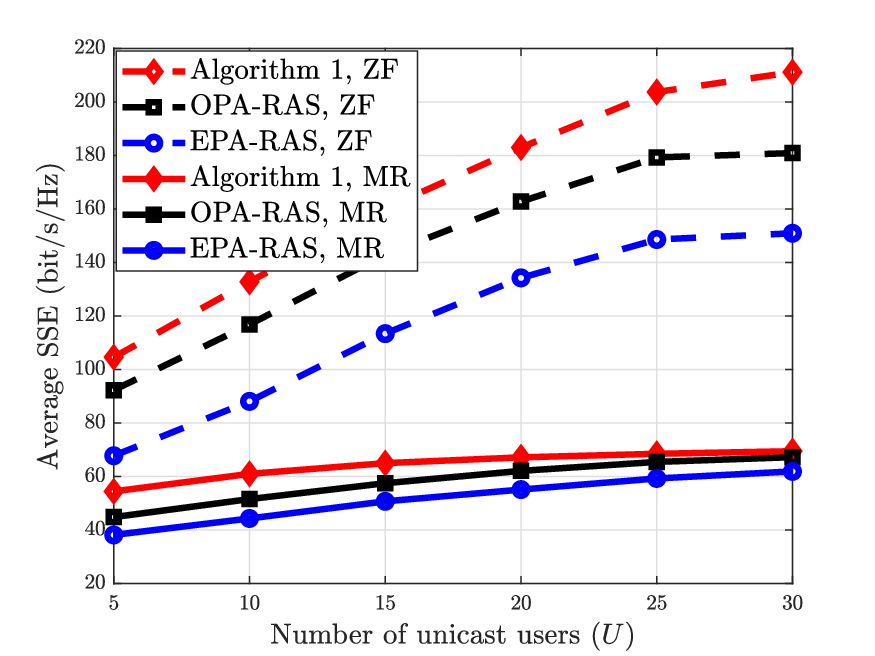}
\centering
\caption{Average  SSE against the number of unicast users,  where  $K_M=12$, $N=60$, $L=36$.}\label{fig5b}
\end{figure}

\begin{figure}[t] 
\centering
\includegraphics[width=0.5\textwidth]{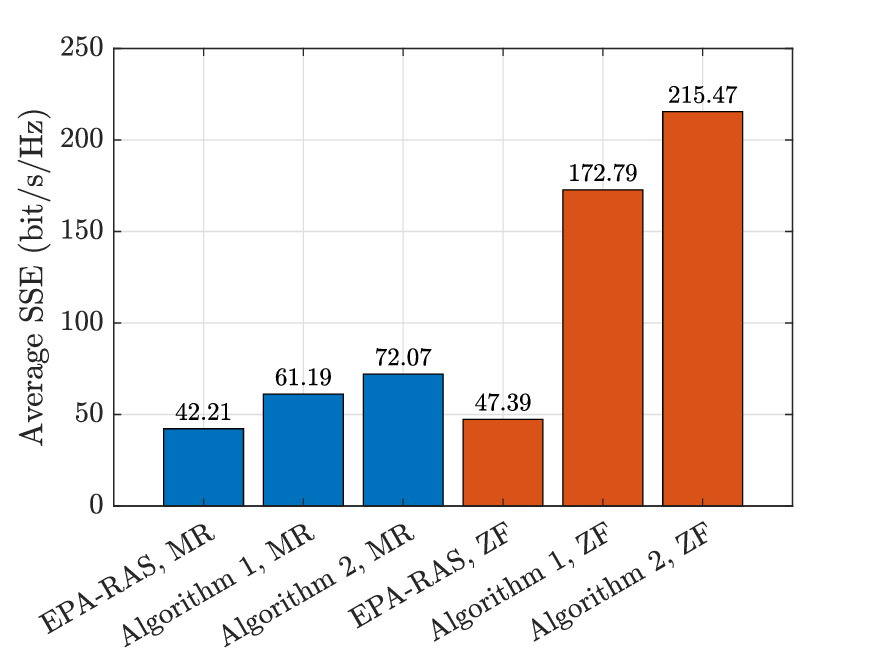}
\centering
\caption{Comparison of the average SSE achieved using \textbf{Algorithm~\ref{algo}} and \textbf{Algorithm \ref{algo2}}, where $U=7$, $K_m=12$, $N=60$.}\label{fig6}
\end{figure}

Figure~\ref{fig2} demonstrates the effectiveness of our proposed APG-based method for joint power allocation and AP selection, as outlined in \textbf{Algorithm~\ref{algo}}, with MR or ZF precoding. We evaluate the cumulative distribution function (CDF) of the SSE achieved by our optimized approach, comparing it against two benchmark schemes: (i) equal power allocation (EPA) with random AP selection (RAS), and (ii) optimized power allocation (OPA) with RAS. The evaluation is performed across different weighting coefficients, $w_1$ and $w_2$. The proposed optimization framework yields significant improvements over both benchmarks. Specifically, in Figure \ref{fig2a}, for $w_1 = 0.2$ and $w_2 = 0.8$, the SSE improves by $39\%$ and $54\%$ compared to the EPA–RAS and OPA–RAS schemes, respectively. Similarly, as shown in Fig. \ref{fig2b}, the gains with ZF precoding reach $53\%$ and up to $620\%$ under the same weighting coefficients. Moreover, the performance advantage of our method becomes more pronounced as $w_2$ increases.

Figure~\ref{fig3} illustrates the impact of the number of APs on the SSE performance of the joint unicast-multicast CF-mMIMO system relying on our proposed \textbf{Algorithm~\ref{algo}}. As the number of APs increases, the system performance improves for both ZF and MR precoding schemes, owing to the higher macro diversity gain—particularly pronounced in the case of MR precoding. Moreover, due to the favorable propagation property, the performance gap between ZF and MR  narrows when a large number of APs are deployed. It is also noteworthy that \textbf{Algorithm~\ref{algo}} maintains strong efficiency across a wide range of  APs configurations in the system. Fig. \ref{fig4} investigates the impact of the number of antennas at each AP on the SSE performance of the joint unicast–multicast CF-mMIMO system based on our proposed \textbf{Algorithm~\ref{algo}}. The results confirm that the proposed joint AP selection and power allocation framework scales effectively with increased antenna counts, preserving high SE.

Figure~\ref{fig5} examines the effect of number of   multicast users on the average  SSE. It is observed that the SSE performance of the joint unicast-multicast CF-mMIMO system significantly improves as the number of multicast users increases relying on  ZF precoding design and our proposed AP selection and power allocation \textbf{Algorithm~\ref{algo}}. However, for MR precoding, when the number of multicast users exceeds $56$, the joint unicast-multicast CF-mMIMO system fails to satisfy the required QoS level ($SE_{QoS}$) for all users. Figure~\ref{fig5b} shows the impact of the number of unicast users on the average SSE. The proposed \textbf{Algorithm~\ref{algo}} provides approximately $60\%$ and $40\%$ SSE gain for the joint unicast–multicast CF-mMIMO system with ZF precoding, when there are $5$ and $30$ unicast users in the system, respectively. 

Figure~\ref{fig6} presents a performance comparison between \textbf{Algorithm~\ref{algo}}, which is APG-based, and \textbf{Algorithm~\ref{algo2}}, which is based on the SCA approach. \textbf{Algorithm~\ref{algo2}} yields a $17.7\%$ performance improvement over \textbf{Algorithm~\ref{algo}} for MR precoding. Furthermore, for ZF precoding, \textbf{Algorithm~\ref{algo2}} achieves a $24.7\%$ improvement relative to \textbf{Algorithm~\ref{algo}}. Note that \textbf{Algorithm~\ref{algo2}} has significantly higher computational complexity than \textbf{Algorithm~\ref{algo}}, especially when the network size is large, as discussed in Section \ref{complex}.

 Finally, Fig.\ref{fig8} presents a case study illustrating the effectiveness of our proposed joint power allocation and AP selection approach. Specifically, the figure shows the SE per user for both the EPA–RAS scheme and \textbf{Algorithm \ref{algo}} under a ZF precoding design in a CF-mMIMO system with $U=3$ unicast users and $M=3$ multicast groups, each containing $K_m=2$ users. In this setup, users labeled $1, 2, 3$ represent unicast users, while those labeled $x_i$ (e.g., $3_2$) represent the $i$-th user in multicast group $x$.
To illustrate the AP–user associations in our case study, we construct the corresponding AP–user association matrix of dimension  $N\times(U+M)$, based on the binary variables $a_{n,u}$ and $\bar{a}_{n,m}$    defined in \eqref{eq:anu}–\eqref{eq:anm}. This matrix represents the serving relationships between APs, unicast users, and multicast groups, where each row corresponds to an AP and each column corresponds to a user or multicast group.
The AP selection matrices corresponding to the EPA–RAS and the proposed \textbf{Algorithm~\ref{algo}}, i.e., $\mathbf{A}_{\text{EPA–RAS}}$ and $\mathbf{A}_{\text{Algorithm~\ref{algo}}}$,  are obtained as
\begin{figure}[t] 
\centering
\includegraphics[width=0.5\textwidth]{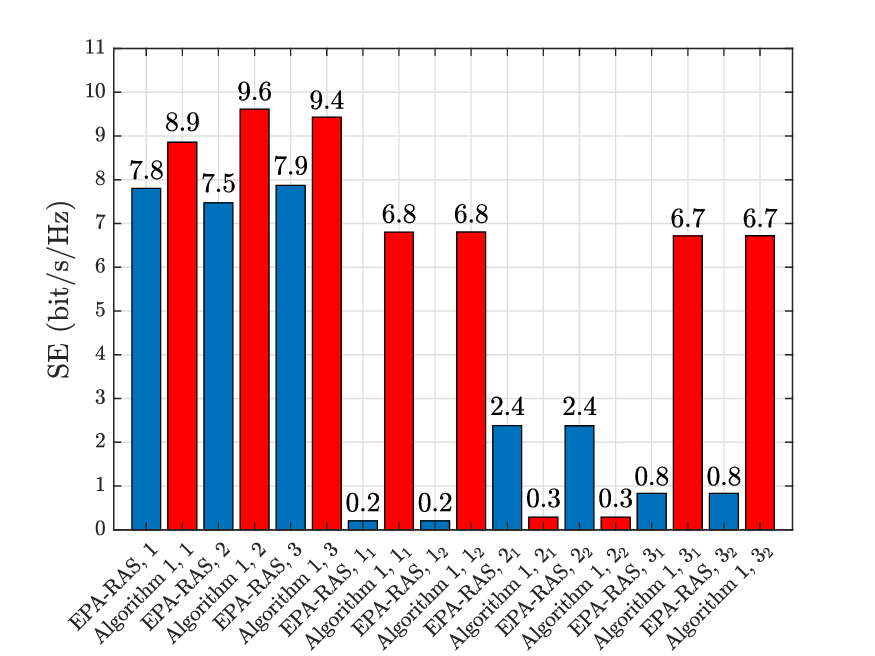}
\centering
\caption{Per-user SE, where $U=3$, $M=3$, $K_m=2$, $N=5$.}\label{fig8}
\end{figure}
\begin{equation}
\mathbf{A}_{\text{EPA–RAS}}=
\begin{bmatrix}
1 & 1 & 1 & 1 & 1 & 1 \\
1 & 0 & 1 & 1 & 0 & 0 \\
0 & 1 & 1 & 0 & 0 & 1 \\
0 & 1 & 0 & 0 & 0 & 0 \\
0 & 0 & 0 & 0 & 1 & 1 \\
\end{bmatrix}, \end{equation}
and
\begin{equation}
\qA_{\text{Algorithm~\ref{algo}}}=
\begin{bmatrix}
1 & 0 & 1 & 0 & 0 & 0 \\
1 & 0 & 1 & 0 & 1 & 0 \\
1 & 1 & 1 & 0 & 0 & 0 \\
1 & 0 & 1 & 0 & 0 & 1 \\
0 & 1 & 1 & 1 & 0 & 0 \\
\end{bmatrix}.
\end{equation}
The total SSE achieved by EPA–RAS is $15$ bit/s/Hz, whereas \textbf{Algorithm~\ref{algo}} achieves $27.7$ bit/s/Hz. It is clear that the minimum SE requirements for both unicast and multicast users are successfully met at the given $\bar{SE}_{QoS}=SE_{QoS}=0.2$ bit/s/Hz.
\section{Conclusion}
We evaluated the performance of a CF-mMIMO system under simultaneous multicast and unicast transmissions, considering both MR and ZF precoding schemes. To maximize the weighted SSE while satisfying strict AP transmit power and user QoS requirements, we proposed a large-scale fading-based joint power allocation and user association optimization framework, leveraging an  APG method. The framework also incorporates practical constraints, such as limits on the number of unicast users and multicast groups per AP, to effectively manage fronthaul overhead.
Numerical results demonstrated that the proposed APG-based approach significantly outperforms conventional heuristic methods in terms of SSE, while achieving faster convergence and considerably lower computational complexity than the SCA-based benchmark. These results highlight the proposed framework’s strong potential for practical implementation in large-scale CF-mMIMO networks, offering an effective and scalable solution for managing mixed unicast and multicast traffic in next-generation wireless systems. Although the focus of this work was on maximizing the SSE, our design choices—such as AP selection and power allocation—also contribute to improved EE. Additionally, the lower complexity of the APG algorithm offers further energy savings. These elements lay the groundwork for future extensions toward joint SE and EE optimization. Moreover, in the context of joint unicast-multicast optimization, Pareto-based multi-objective optimization, by identifying optimal trade-offs between unicast and multicast performance, represents a promising direction for future research.

\appendices
\section{Proof of Proposition 1}
\label{proumr}
The numerator of \eqref{eq:sinru} is computed as
\vspace{-0.5em}
\begin{align} \label{eq:dsub}
\text{DS}_u^{\varsigma}&=\sqrt{p_{\text{dl}}} ~\mathbb{E} \left\{\sum_{n \in \mathcal{N}} a_{n,u} \sqrt{\eta_{n,u}} \mathbf{c}_{n,u}^H \boldsymbol{b}_{n,u}^\varsigma \right\} 
\nonumber\\
& = \sqrt{p_{\text{dl}}}\sum_{n \in \mathcal{N}} a_{n,u} \sqrt{\eta_{n,u}} ~\mathbb{E} \left\{ \mathbf{\hat{c}}_{n,u}^H \boldsymbol{b}_{n,u}^\varsigma \right\}.
\end{align} 
For a given precoding scheme, i.e., MR and ZF, substituting \eqref{eq:bnu} into \eqref{eq:dsub} yields
\begin{align} \label{eq:dsua} 
\text{DS}_u^{\varsigma} = 
\begin{cases}    \sqrt{p_{\text{dl}}}L\sum\limits_{n \in \mathcal{N}} a_{n,u} \sqrt{\eta_{n,u}} \mathbf\gamma_{n,u}, & \varsigma = \MR,
    \\
\sqrt{p_{\text{dl}}}\sum\limits_{n \in \mathcal{N}} a_{n,u} \sqrt{\eta_{n,u}}, & \varsigma =\ZF,
\end{cases}
\end{align}
where \eqref{eq:dsub} is derived based on the uncorrelated nature of the channel estimation error $\mathbf{\Tilde{c}}_{n,u}$ and the channel estimate $\mathbf{\hat{c}}_{n,u}$. For the computation of \eqref{eq:dsua}  with MR, we have exploited $\mathbb{E} \big\{ \mathbf{\hat{c}}_{n,u}^H \mathbf{\hat{c}}_{n,u} \big\}=L\mathbf\gamma_{n,u}$.
On the other hand, the first term in the denominator of \eqref{eq:sinru} is calculated as
\begin{align}
 \label{eq:buua}
 &  \mathbb{E} \left\{\left| \text{BU}_u^{\varsigma} \right|^2\right\}=\nonumber\\
 &p_{\text{dl}}  \sum_{n \in \mathcal{N}} a_{n,u} \eta_{n,u} \Big(\mathbb{E}\left\{\left|\mathbf{c}_{n,u}^H \boldsymbol{b}_{n,u}^\varsigma - \mathbb{E} \left\{ \mathbf{c}_{n,u}^H\boldsymbol{b}_{n,u}^\varsigma \right\}\right|^2 \right\}\Big) 
\nonumber\\
&= p_{\text{dl}}  \sum_{n \in \mathcal{N}} a_{n,u} \eta_{n,u} \Big(\!\mathbb{E}\left\{\left|\mathbf{\Tilde{c}}_{n,u}^H \boldsymbol{b}_{n,u}^\varsigma\right|^2\right\}  +\mathbb{E}\left\{\left|\mathbf{\hat{c}}_{n,u}^H \boldsymbol{b}_{n,u}^\varsigma\right|^2\right\}\nonumber\\
&\hspace{14em}-\left|\mathbb{E} \left\{\mathbf{c}_{n,u}^H  \boldsymbol{b}_{n,u}^\varsigma \right\}\right|^2 \Big)
\nonumber\\
& =
\begin{cases}
     p_{\text{dl}} L\sum\limits_{n \in \mathcal{N}} a_{n,u} \eta_{n,u}  \gamma_{n,u} \beta_{n,u}, & \varsigma = \MR,
    \\
   p_{\text{dl}}  \sum\limits_{n \in \mathcal{N}}  a_{n,u} \eta_{n,u} \frac{(\beta_{n,u}-\gamma_{n,u})}{(L-U-M)\mathbf\gamma_{n,u}},  & \varsigma = \ZF,
\end{cases} 
\end{align}
where we have exploited $\mathbb{E}\left\{|\mathbf{\Tilde{c}}_{n,u}^H \mathbf{\hat{c}}_{n,u} |^2\right\}=L \gamma_{n,u} (\beta_{n,u}-\gamma_{n,u})$ and $\mathbb{E}\big\{\| \mathbf{\hat{c}}_{n,u}\|^4\big\}=L(L+1) \gamma_{n,u}^2$ for the MR scheme. Additionally, for the ZF scheme, we have utilized the central complex Wishart matrix identity, $\mathbb{E} \big\{ (\boldsymbol{b}_{n,u}^\ZF)^H\boldsymbol{b}_{n,u}^\ZF \big\}=\frac{1}{(L-U-M)\mathbf\gamma_{n,u}}$, as given in \cite[Lemma 2.10]{tulino04}. Thus, to implement ZF, the condition $L > U+M$ must be satisfied. Similarly,  $\mathbb{E} \big\{| \text{UI}_{u,u'}^\varsigma |^2\big\}$ and $\mathbb{E} \big\{| \text{MI}_{u,m}^\varsigma |^2\big\}$ can be computed, respectively, as
\vspace{-0.5em}
\begin{align} \label{eq:uiuua}
&\mathbb{E} \left\{\left| \text{UI}_{u,u'}^\varsigma \right|^2\right\}=p_{\text{dl}} \mathbb{E} \left\{\left| \sum_{n \in \mathcal{N}} a_{n,u'} \sqrt{\eta_{n,u'}} \mathbf{c}_{n,u}^H  \boldsymbol{b}_{n,u'}^\varsigma \right|^2 \right\} 
\nonumber\\ 
&=  p_{\text{dl}}  \sum_{n \in \mathcal{N}} a_{n,u'} \eta_{n,u'} \bigg(\mathbb{E}\left\{\left|\mathbf{\Tilde{c}}_{n,u}^H\boldsymbol{b}_{n,u'}^\varsigma \right|^2 +\left|\mathbf{\hat{c}}_{n,u}^H\boldsymbol{b}_{n,u'}^\varsigma\right|^2\right\}  \bigg) \nonumber\\
&=\!\!
    \begin{cases}
        \!p_{\text{dl}} L\!\!\sum\limits_{n \in \mathcal{N}} a_{n,u'} \eta_{n,u'} \beta_{n,u} \gamma_{n,u'}, & \varsigma \!= \MR,
        \\
         \!p_{\text{dl}}  \!\!\sum\limits_{n \in \mathcal{N}} a_{n,u'} \eta_{n,u'} \frac{(\beta_{n,u}-\gamma_{n,u})}{(L-U-M)\mathbf\gamma_{n,u'}}, & \varsigma \!= \ZF,
    \end{cases}
\end{align}
and
\vspace{-0.5em}
\begin{align}
\label{eq:miuma}
 & \mathbb{E} \left\{\left| \text{MI}_{u,m}^\varsigma \right|^2\right\}=p_{\text{dl}} \mathbb{E} \left\{\left|\sum_{n \in \mathcal{N}} \bar{a}_{n,m} \sqrt{\bar{\eta}_{n,m}} \mathbf{c}_{n,u}^H  \bar{\boldsymbol{b}}_{n,m}^\varsigma \right|^2 \right\}
\nonumber\\ 
&=  p_{\text{dl}}  \sum_{n \in \mathcal{N}} \bar{a}_{n,m} \bar{\eta}_{n,m} \bigg(\mathbb{E}\left\{\left|\mathbf{\Tilde{c}}_{n,u}^H \bar{\boldsymbol{b}}_{n,m}^\varsigma \right|^2 +\left|\mathbf{\hat{c}}_{n,u}^H \bar{\boldsymbol{b}}_{n,m}^\varsigma \right|^2 \right\}  \bigg)\nonumber\\
&=
    \begin{cases}
        p_{\text{dl}} L \sum\limits_{n \in \mathcal{N}} \bar{a}_{n,m} \bar{\eta}_{n,m}  \beta_{n,u} \zeta_{n,m}, & \varsigma=\MR,
        \\
         p_{\text{dl}} \sum\limits_{n \in \mathcal{N}} \bar{a}_{n,m} \bar{\eta}_{n,m}  \frac{(\beta_{n,u}-\gamma_{n,u})}{(L-U-M)\mathbf\zeta_{n,m}}, & \varsigma=\ZF.
    \end{cases}
\end{align}
Finally, by substituting  \eqref{eq:dsua}, \eqref{eq:buua}, \eqref{eq:uiuua}, and \eqref{eq:miuma} into \eqref{eq:sinru}, $\mathrm{SINR}_u^\MR$ and $\mathrm{SINR}_u^\ZF$ at the $u$-th unicast user in \eqref{eq:sinrumra} and \eqref{eq:sinruzfa} can be obtained.
\vspace{-0.6em}
\section{Proof of Proposition 2} \label{promkmr}
The numerator of \eqref{eq:sinrmk} is computed as
\begin{align}
\label{eq:dsmkb}
\text{DS}_{m,k}^\varsigma&=\sqrt{p_{\text{dl}}} ~\mathbb{E} \left\{\sum_{n \in \mathcal{N}}
\bar{a}_{n,m} \sqrt{\bar{\eta}_{n,m}} \mathbf{t}_{n,m,k}^H   \bar{\boldsymbol{b}}_{n,m}^\varsigma \right\}
\nonumber\\
&=\!\sqrt{p_{\text{dl}}} \!\sum_{n \in \mathcal{N}} \bar{a}_{n,m} \sqrt{\bar{\eta}_{n,m}} \mathbb{E}\left\{\mathbf{\hat{t}}_{n,m,k}^H\bar{\boldsymbol{b}}_{n,m}^\varsigma \right\}.
\end{align}
For a given precoding scheme, i.e., MR and ZF, substituting \eqref{eq:bnm} into \eqref{eq:dsmkb} yields
\begin{align}  \label{eq:dsmka}
\text{DS}_{m,k}^\varsigma=
    \begin{cases}
        \!\sqrt{p_{\text{dl}}}L\!\!\!\sum\limits_{n \in \mathcal{N}} \!\bar{a}_{n,m} \sqrt{\bar{\eta}_{n,m}} \sqrt{\zeta_{n,m}\bar{\mathbf\gamma}_{n,m,k}}, & \varsigma=\MR,
        \\
         \!\sqrt{p_{\text{dl}}}\sum\limits_{n \in \mathcal{N}} \bar{a}_{n,m} \sqrt{\bar{\eta}_{n,m}}, & \varsigma=\ZF,
    \end{cases}
\end{align} 
where \eqref{eq:dsmkb} is derived based on the uncorrelated nature of the channel estimation error $\mathbf{\Tilde{t}}_{n,m,k}$ and the multicast group channel estimation $\mathbf{\hat{t}}_{n,m}$. Moreover,   \eqref{eq:dsmka} follows from the fact that $\mathbb{E} \big\{ \mathbf{\hat{t}}_{n,m,k}^H \mathbf{\hat{t}}_{n,m} \big\}=L\sqrt{\zeta_{n,m}\bar{\mathbf\gamma}_{n,m,k}}$.

In addition, the first term in the denominator of \eqref{eq:sinrmk} is calculated as
\begin{align}
\label{eq:bumka}
 &\mathbb{E} \left\{\left|\text{BU}_{m,k}^\varsigma \right|^2 \right\}=\nonumber\\
 &p_{\text{dl}} \sum_{n \in \mathcal{N}} \bar{a}_{n,m}  \bar{\eta}_{n,m} \bigg(\mathbb{E}\left\{\left| \mathbf{t}_{n,m,k}^H \bar{\boldsymbol{b}}_{n,m}^\varsigma - \mathbb{E} \left\{ \mathbf{t}_{n,m,k}^H \bar{\boldsymbol{b}}_{n,m}^\varsigma \right\} \right|^2 \right\}\bigg)
 \nonumber\\
 &=\! p_{\text{dl}} \! \sum_{n \in \mathcal{N}} \!\!\bar{a}_{n,m}  \bar{\eta}_{n,m} \!\bigg(\!\!\mathbb{E}\left\{\left|\mathbf{\Tilde{t}}_{n,m,k}^H \bar{\boldsymbol{b}}_{n,m}^\varsigma\right|^2\right\}+\mathbb{E}\left\{\left|\mathbf{\hat{t}}_{n,m,k}^H\bar{\boldsymbol{b}}_{n,m}^\varsigma\right|^2\right\} \nonumber\\
 & \hspace{12 em} - \left|\mathbb{E} \left\{  \mathbf{t}_{n,m,k}^H \bar{\boldsymbol{b}}_{n,m}^\varsigma \right\}\right|^2 \bigg)\nonumber\\
& =\!
    \begin{cases}
        \!p_{\text{dl}}  L\!\!\sum\limits_{n \in \mathcal{N}} \!\!\bar{a}_{n,m} \bar{\eta}_{n,m}\bar{\beta}_{n,m,k}\zeta_{n,m}, &\!\!\varsigma\!=\!\MR,
        \\
        \!p_{\text{dl}} \!\! \sum\limits_{n \in \mathcal{N}} \!\!\bar{a}_{n,m}  \bar{\eta}_{n,m} \frac{\bar{\beta}_{n,m,k}-\bar{\mathbf\gamma}_{n,m,k}}{(L-U-M)\mathbf\zeta_{n,m}}, &\!\!\varsigma\!=\!\ZF.
    \end{cases}
\end{align}
The calculation in \eqref{eq:bumka} for the MR scheme is based on the fact that $\mathbb{E}\big\{|\mathbf{\Tilde{t}}_{n,m,k}^H \mathbf{\hat{t}}_{n,m}|^2\big\}=L~\zeta_{n,m} (\bar{\beta}_{n,m,k}-\bar{\mathbf\gamma}_{n,m,k})$ and $\mathbb{E}\big\{|\mathbf{\hat{t}}_{n,m,k}^H\mathbf{\hat{t}}_{n,m}|^2\big\}=L(L+1) \bar{\mathbf\gamma}_{n,m,k}\zeta_{n,m}$, while for the ZF scheme, we use  $\mathbb{E} \left\{ (\bar{\boldsymbol{b}}_{n,m}^\ZF)^H\bar{\boldsymbol{b}}_{n,m}^\ZF \right\}=\frac{1}{(L-U-M)\mathbf\zeta_{n,m}}$.
Following similar steps,  $\mathbb{E} \big\{|\text{MI}_{m,k,m'}^\varsigma |^2\big\}$ and  $\mathbb{E} \big\{|\text{UI}_{m,k,u}^\varsigma |^2\big\}$ can be calculated, respectively, as
\begin{align}
\label{eq:mimkma}
&\mathbb{E} \left\{\left|\text{MI}_{m,k,m'}^\varsigma \right|^2\!\right\}\!=\!p_{\text{dl}} \mathbb{E} \left\{\left| \sum_{n \in \mathcal{N}} \bar{a}_{n,m'} \sqrt{\bar{\eta}_{n,m'}} \mathbf{t}_{n,m,k}^H  \bar{\boldsymbol{b}}_{n,m'}^\varsigma \right|^2 \right\}
\nonumber\\
&=\! p_{\text{dl}} \!\! \sum_{n \in \mathcal{N}} \!\!\bar{a}_{n,m'}  \bar{\eta}_{n,m'} \bigg(\mathbb{E}\left\{\left|\mathbf{\Tilde{t}}_{n,m,k}^H \bar{\boldsymbol{b}}_{n,m'}^\varsigma\right|^2 + \left|\mathbf{\hat{t}}_{n,m,k}^H\bar{\boldsymbol{b}}_{n,m'}^\varsigma \right|^2\right\}\bigg) 
\nonumber\\ 
&=
    \begin{cases}
        p_{\text{dl}} L\sum\limits_{n \in \mathcal{N}} \bar{a}_{n,m'} \bar{\eta}_{n,m'} \bar{\beta}_{n,m,k}  \zeta_{n,m'}, & \varsigma=\MR,
        \\
        p_{\text{dl}}  \sum\limits_{n \in \mathcal{N}} \bar{a}_{n,m'}  \bar{\eta}_{n,m'} \frac{\bar{\beta}_{n,m,k}-\bar{\mathbf\gamma}_{n,m,k}}{(L-U-M)\mathbf\zeta_{n,m'}}, & \varsigma=\ZF,
    \end{cases}
\end{align}
and
\vspace{-0.5em}
\begin{align}
\label{eq:uimkua}
& \mathbb{E} \left\{\left|\text{UI}_{m,k,u}^\varsigma \right|^2\right\} = p_{\text{dl}} \mathbb{E} \left\{\left| \sum_{n \in \mathcal{N}} a_{n,u} \sqrt{\eta_{n,u}} \mathbf{t}_{n,m,k}^H  \boldsymbol{b}_{n,u}^\varsigma \right|^2 \right\} 
\nonumber\\
&= p_{\text{dl}}  \sum_{n \in \mathcal{N}} a_{n,u}  \eta_{n,u} \bigg(\mathbb{E}\left\{\left|\mathbf{\Tilde{t}}_{n,m,k}^H \boldsymbol{b}_{n,u}^\varsigma  + \mathbf{\hat{t}}_{n,m,k}^H \boldsymbol{b}_{n,u}^\varsigma \right|^2\right\} \bigg) 
\nonumber\\
& =
    \begin{cases}
        p_{\text{dl}} ~L\sum\limits_{n \in \mathcal{N}} a_{n,u} \eta_{n,u} \bar{\beta}_{n,m,k} \gamma_{n,u}, & \varsigma=\MR,
        \\
        p_{\text{dl}}  \sum\limits_{n \in \mathcal{N}} a_{n,u}  \eta_{n,u} \frac{\bar{\beta}_{n,m,k}-\bar{\mathbf\gamma}_{n,m,k}}{(L-U-M)\mathbf\gamma_{n,u}}, & \varsigma=\ZF.
    \end{cases}
\end{align}
Plugging \eqref{eq:dsmka}, \eqref{eq:bumka}, \eqref{eq:mimkma}, and \eqref{eq:uimkua} into \eqref{eq:sinrmk},  gives the $\mathrm{SINR}_{m,k}^\MR$ and $\mathrm{SINR}_{m,k}^\ZF$ at the $k_m$-th multicast user in \eqref{eq:sinrmkmra} and \eqref{eq:sinrmkzfa}, respectively. 
\bibliographystyle{IEEEtran}
\bibliography{bibliography}

\end{document}